\documentclass[nonacm]{acmart}

\usepackage{microtype}
\usepackage{xcolor}
\usepackage{framed}
\usepackage{booktabs}
\usepackage{tikz-cd}
\usepackage[unicode]{hyperref}
\usepackage{doi}

\newcommand{\namedref}[2]{\hyperref[#2]{#1~\ref*{#2}}}
\newcommand{\sectionref}[1]{\namedref{Section}{#1}}

\newcommand{\theoremref}[1]{\namedref{Theorem}{#1}}
\newcommand{\corollaryref}[1]{\namedref{Corollary}{#1}}
\newcommand{\figureref}[1]{\namedref{Figure}{#1}}
\newcommand{\lemmaref}[1]{\namedref{Lemma}{#1}}
\newcommand{\tableref}[1]{\namedref{Table}{#1}}
\newcommand{\defref}[1]{\namedref{Definition}{#1}}

\newcommand{\s}{\mspace{1mu}}
\newcommand{\mybox}[1]{\mspace{2mu}{\setlength{\fboxsep}{1.5pt}\color{lightgray}\boxed{\color{black}\scriptstyle #1}}\mspace{2mu}}
\newcommand{\M}{\mathsf{M}}
\renewcommand{\P}{\mathsf{P}}
\renewcommand{\O}{\mathsf{O}}
\newcommand{\X}{\mathsf{X}}

\newcommand{\bX}{\mybox{\X}}
\newcommand{\bMX}{\mybox{\M\X}}

\newcommand{\bOX}{\mybox{\O\X}}

\newcommand{\bPOX}{\mybox{\P\O\X}}
\newcommand{\bMOX}{\mybox{\M\O\X}}
\newcommand{\bMPOX}{\mybox{\M\P\O\X}}

\newtheorem{theorem}{Theorem}[section]
\newtheorem{lemma}[theorem]{Lemma}
\newtheorem{corollary}[theorem]{Corollary}

\theoremstyle{definition}
\newtheorem{definition}{Definition}[section]

\DeclareMathOperator{\poly}{poly}

\begin{document}
\title{Lower Bounds for Maximal Matchings and Maximal Independent Sets}

\author{Alkida Balliu}
\affiliation{%
  \institution{Gran Sasso Science Institute}
  \city{L'Aquila}
  \country{Italy}}
\email{alkida.balliu@gssi.it}

\author{Sebastian Brandt}
\affiliation{%
  \institution{CISPA Helmholtz Center for Information Security}
  \city{Saarbr\"ucken}
  \country{Germany}}
\email{brandt@cispa.de}

\author{Juho Hirvonen}
\affiliation{%
  \institution{Aalto University}
  \city{Helsinki}
  \country{Finland}}
\email{juho.hirvonen@aalto.fi}

\author{Dennis Olivetti}
\affiliation{%
  \institution{Gran Sasso Science Institute}
  \city{L'Aquila}
  \country{Italy}}
\email{dennis.olivetti@gssi.it}

\author{Mika\"el Rabie}
\affiliation{%
  \institution{Aalto University}
  \city{Helsinki}
  \country{Finland}}
\affiliation{%
  \institution{Universit\'e de Paris}
  \city{Paris}
  \country{France}}
\email{mikael.rabie@irif.fr}

\author{Jukka Suomela}
\affiliation{%
  \institution{Aalto University}
  \city{Helsinki}
  \country{Finland}}
\email{jukka.suomela@aalto.fi}

\begin{abstract}
There are distributed graph algorithms for finding maximal matchings and maximal independent sets in $O(\Delta + \log^* n)$ communication rounds; here $n$ is the number of nodes and $\Delta$ is the maximum degree. The lower bound by Linial (1987, 1992) shows that the dependency on $n$ is optimal: these problems cannot be solved in $o(\log^* n)$ rounds even if $\Delta = 2$. However, the dependency on $\Delta$ is a long-standing open question, and there is currently an exponential gap between the upper and lower bounds.

We prove that the upper bounds are tight. We show that any algorithm that finds a maximal matching or maximal independent set with probability at least $1-1/n$ requires $\Omega(\min\{\Delta,\log \log n / \log \log \log n\})$ rounds in the LOCAL model of distributed computing. As a corollary, it follows that any deterministic algorithm that finds a maximal matching or maximal independent set requires $\Omega(\min\{\Delta, \log n / \log \log n\})$ rounds; this is an improvement over prior lower bounds also as a function of~$n$.
\end{abstract}

\ccsdesc[500]{Theory of computation~Distributed computing models}
\ccsdesc[300]{Theory of computation~Complexity classes}

\keywords{Maximal matching, maximal independent set, distributed graph algorithms, lower bounds}

\maketitle

\section{Introduction}

There are four classic problems that have been studied extensively in distributed graph algorithms since the very beginning of the field in the 1980s~\cite{Ghaffari2017b}: \emph{maximal independent set} (MIS), \emph{maximal matching} (MM), \emph{vertex coloring with $\Delta+1$ colors}, and \emph{edge coloring with $2\Delta-1$ colors}; here $\Delta$ is the maximum degree of the graph. All of these problems are trivial to solve with a greedy centralized algorithm, but their distributed computational complexity has remained an open question.

In this work, we resolve the distributed complexity of MIS and MM in the region $\Delta \ll \log \log n$. In this region, for the LOCAL model~\cite{Linial1992,Peleg2000} of distributed computing, the fastest known algorithms for these problems are:
\begin{itemize}
    \item MM is possible in $O(\Delta + \log^* n)$ rounds~\cite{panconesi01simple}.
    \item MIS is possible in $O(\Delta + \log^* n)$ rounds~\cite{barenboim14distributed}.
\end{itemize}
Nowadays we know how to find a vertex or edge coloring with $O(\Delta)$ colors in $o(\Delta) + O(\log^* n)$ rounds~\cite{fraigniaud16local,barenboim16sublinear}. Hence the current algorithms for both MIS and MM are conceptually very simple: color the vertices or edges with $O(\Delta)$ colors, and then construct an independent set or matching by going through all color classes one by one (first select all nodes or edges of color 1, then select the nodes or edges of color 2 that are not adjacent to previously selected nodes or edges, then do the same for color 3, etc.). The second part is responsible for the $O(\Delta)$ term in the running time, and previously we had no idea if this is necessary.

\begin{table}
    \centering
    \caption{Efficient algorithms for MM and MIS.}\label{tab:related}
    \begin{tabular}{lllll}
        \toprule
        Problem & $\Delta$ & Randomness & Complexity & Reference \\
        \midrule
        MM & small & deterministic & $O(\Delta + \log^* n)$ & \citet{panconesi01simple} \\
        & large & deterministic & $O(\log^2\Delta \cdot \log n)$ & \citet{fischer17improved} \\
        & large & randomized & $O(\log \Delta + \log^3 \log n)$ & \citet{fischer17improved,Barenboim2012,Barenboim2016} \\
        \midrule
        MIS & small & deterministic & $O(\Delta + \log^* n)$ & \citet{barenboim14distributed} \\
        & large & deterministic & $O(\log^7 n)$ & \citet{Rozhon2020} \\
        & large & randomized & $O(\log \Delta + \log^7\log n)$ & \citet{Rozhon2020} \\
        \bottomrule
    \end{tabular}
\end{table}

\subsection{Prior work}

Already in the 1990s we had a complete understanding of the term $\log^* n$:
\begin{itemize}
    \item MM is not possible in $f(\Delta) + o(\log^* n)$ rounds for any $f$~\cite{Linial1987,Linial1992,Naor1991}.
    \item MIS is not possible in $f(\Delta) + o(\log^* n)$ rounds for any $f$~\cite{Linial1987,Linial1992,Naor1991}.
\end{itemize}
Here the upper bounds are deterministic and the lower bounds hold also for randomized algorithms.

However, we have had no lower bounds that would exclude the existence of algorithms of complexity $o(\Delta) + O(\log^* n)$ for either of these problems~\cite{Barenboim2013,Suomela2014}. For regular graphs we have not even been able to exclude the possibility of solving both of these problems in time $O(\log^* n)$, while for the general case the best lower bound as a function of $\Delta$ was $\Omega(\log \Delta / \log \log \Delta)$~\cite{Kuhn2004,Kuhn2006,kuhn16local,Coupette2020}.

\subsection{Contributions}

We close the gap and prove that the current upper bounds are tight. There is no algorithm of complexity $o(\Delta) + O(\log^* n)$ for MM or MIS. More precisely, our main result is:
\begin{center}
\parbox{112mm}{
\begin{framed}
    \begin{center}
    Any algorithm that solves MM or MIS with probability at least $1-1/n$ requires\\[2pt]
    $\displaystyle \Omega\Bigl(\min\Bigl\{\Delta, \frac{\log \log n}{\log \log \log n}\Bigr\}\Bigr)$
    rounds in the LOCAL model.

    \vspace{6mm}
    Any deterministic algorithm that solves MM or MIS requires\\[2pt]
    $\displaystyle \Omega\Bigl(\min\Bigl\{\Delta, \frac{\log n}{\log \log n}\Bigr\}\Bigr)$
    rounds in the LOCAL model.
    \end{center}
\end{framed}
}
\end{center}

\noindent
As corollaries, we have a new separation and a new equivalence in the $\Delta \ll \log \log n$ region (see \sectionref{ssec:related-coloring}):
\begin{itemize}
    \item MM and MIS are strictly harder than $(\Delta+1)$-vertex coloring and $(2\Delta-1)$-edge coloring.
    \item MM and MIS are exactly as hard as greedy coloring.
\end{itemize}

\begin{figure}
    \centering
    \includegraphics[page=7,scale=0.7]{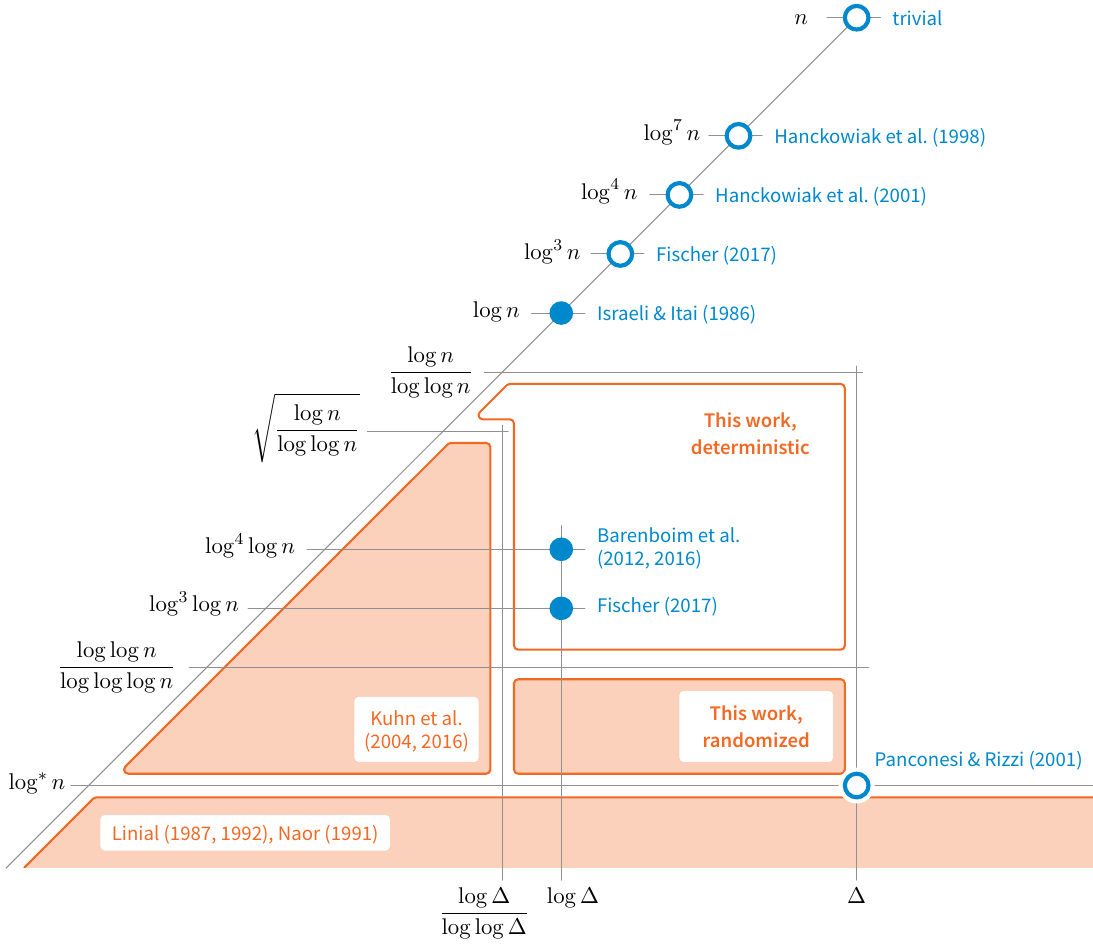}
    \Description{}
    \caption{Distributed algorithms for maximal matching: upper bounds (blue dots) and lower bounds (orange regions). Filled dots are randomized algorithms and filled regions are lower bounds for randomized algorithms; white dots are deterministic algorithms and white regions are lower bounds for deterministic algorithms. The running time is represented here in the form $O(f(\Delta) + g(n))$, the horizontal axis represents the $f(\Delta)$ term, and the vertical axis represents the $g(n)$ term.}\label{fig:mm-bounds}
\end{figure}

\subsection{Plan}

We will present a simpler version of our new linear-in-$\Delta$ lower bound in \sectionref{sec:deterministic}. There we will introduce a restricted setting of distributed computing---deterministic algorithms in the port-numbering model---and explain the key ideas using that. In \sectionref{sec:randomized} we will then see how to extend the result to randomized and deterministic algorithms in the usual LOCAL model of computing.

\section{Related work}\label{sec:related}

\subsection{MIS and MM}

For MIS and MM, as well as for other classical symmetry-breaking problems, there are three major families of algorithms:
\begin{itemize}
    \item Deterministic algorithms suitable for small $\Delta$, with a complexity of the form $f(\Delta) + O(\log^* n)$.
    \item Deterministic algorithms suitable also for large $\Delta$, with superlogarithmic complexities as a function of $n$.
    \item Randomized algorithms suitable also for large $\Delta$, with at most logarithmic complexities as a function of $n$.
\end{itemize}
We summarize the state of the art in \tableref{tab:related} and \figureref{fig:mm-bounds}; see e.g.\ \citet{Alon1986}, \citet{Luby1985,Luby1986}, \citet{Israeli1986}, \citet{panconesi96decomposition}, \citet{Hanckowiak2001,Hanckowiak1998}, \citet{Barenboim2012,Barenboim2016}, and \citet{ghaffari16improved} for more prior work on maximal matchings and maximal independent sets.
The listed upper bounds for MIS by Rozho\v{n} and Ghaffari~\cite{Rozhon2020} were obtained after the preliminary conference version of our work.

Previously, it was not known if any of these algorithms are optimal. In essence, there have been only two lower bound results:
\begin{itemize}
    \item \citet{Linial1987,Linial1992} and \citet{Naor1991} show that any deterministic or randomized algorithm for MM or MIS requires $\Omega(\log^* n)$ rounds, even if we have $\Delta = 2$.
    \item \citet{Kuhn2004,Kuhn2006,kuhn16local} show that any deterministic or randomized algorithm for MM or MIS requires \[\Omega\biggl(\min\biggl\{\frac{\log \Delta}{\log \log \Delta}, \sqrt{\frac{\log n}{\log \log n}}\biggr\}\biggr)\] rounds.
\end{itemize}
Hence, for example, when we look at the fastest MM algorithms, dependency on $n$ in $O(\Delta + \log^* n)$ is optimal, and dependency on $\Delta$ in $O(\log \Delta + \log^3 \log n)$ is near-optimal. However, could we get the best of both worlds and solve MM or MIS in e.g.\ $O(\log\Delta + \log^* n)$ rounds?

In this work we show that the answer is no. There is no algorithm that runs in time $o(\Delta) + O(\log^* n)$. The current upper bounds for the case of a small $\Delta$ are optimal. Moreover, our work shows there is not much room for improvement in the dependency on $n$ in the algorithm by \citet{fischer17improved}, either.

With this result we resolve Open Problem 11.6 in the manuscript of \citet{Barenboim2013}, and present a proof for the conjecture of \citet{Goos2017}.

\subsection{Coloring}\label{ssec:related-coloring}

It is interesting to compare MIS and MM with the classical distributed coloring problems: vertex coloring with $\Delta+1$ colors and edge coloring with $2\Delta-1$ colors~\cite{Barenboim2013}. As recently as in 2014, the fastest algorithms for all of these problems in the ``small $\Delta$'' region had the same complexity as MIS and MM, $O(\Delta + \log^* n)$ rounds~\cite{barenboim14distributed}. 
However, in 2015 the paths diverged: \citet{barenboim16sublinear} and \citet{fraigniaud16local} have presented algorithms for graph coloring in $o(\Delta) + O(\log^* n)$ rounds, and hence we now know that coloring is strictly easier than MM or MIS in the small $\Delta$ region.

The only known lower bound for $(\Delta+1)$-vertex coloring and $(2\Delta-1)$-edge coloring is $\Omega(\log^* n)$~\cite{Linial1987,Linial1992,Naor1991}; better lower bounds are known only for coloring algorithms of a certain form, called locally-iterative algorithms~\cite{szegedy1993locality,kuhn2006complexity}.

However, there is a variant of $(\Delta+1)$-vertex coloring that is closely related to MIS: \emph{greedy coloring}~\cite{Gavoille2009}. Greedy coloring is trivially at least as hard as MIS, as color class $1$ in any greedy coloring gives an MIS. On the other hand, greedy coloring is possible in time $O(\Delta + \log^* n)$, as we can turn an $O(\Delta)$-vertex coloring into a greedy coloring in $O(\Delta)$ rounds (and this was actually already known to be tight). Now our work shows that greedy coloring is exactly as hard as MIS. In a sense this is counterintuitive: finding just color class $1$ of a greedy coloring is already asymptotically as hard as finding the entire greedy coloring.

\subsection{Restricted lower bounds}

While no linear-in-$\Delta$ lower bounds for MM or MIS were known previously for the usual LOCAL model of distributed computing, there was a tight bound for a toy model of distributed computing: deterministic algorithms in the \emph{edge coloring model}. Here we are given a proper edge coloring of the graph with $\Delta+1$ colors, the nodes are anonymous, and the nodes can use the edge colors to refer to their neighbors. In this setting there is a trivial algorithm that finds an MM in $O(\Delta)$ rounds: go through color classes one by one and greedily add edges that are not adjacent to anything added so far. It turns out there is a matching lower bound: no algorithm solves this problem in $o(\Delta)$ rounds in the same model~\cite{Hirvonen2012}.

The same technique was later used to study \emph{maximal fractional matchings} in the LOCAL model of computing. This is a problem that can be solved in $O(\Delta)$ rounds (independent of $n$) in the usual LOCAL model~\cite{Astrand2010}, and there was a matching lower bound that shows that the same problem cannot be solved in $o(\Delta)$ rounds (independent of $n$) in the same model~\cite{Goos2017}.

While these lower bounds were seen as a promising indicator that there might be a linear-in-$\Delta$ lower bound in a more general setting, the previous techniques turned out to be a dead end. In particular, they did not tell anything nontrivial about the complexity of MM or MIS in the usual LOCAL model. Now we know that an entirely different kind of approach was needed---even though the present work shares some coauthors with~\cite{Hirvonen2012,Goos2017}, the techniques of the present work are entirely unrelated to those.

\subsection{Speedup simulation technique}

The technique that we use in this work is based on \emph{speedup simulation}, a.k.a., \emph{round elimination}. In essence, the idea is that we assume we have an algorithm $A$ that solves a problem $\Pi$ in $T$ rounds, and then we construct a new algorithm $A'$ that solves another problem $\Pi'$ in $T' \le T-1$ rounds. A node in algorithm $A'$ gathers its radius-$T'$ neighborhood, considers all possible ways of extending it to a radius-$T$ neighborhood, simulates $A$ for each such extension, and then uses the output of $A$ to choose its own output. Now if we can iterate the speedup step for $k$ times (without reaching a trivial problem), we know that the original problem requires at least $k$ rounds to solve.

This approach was first used by \citet{Linial1987,Linial1992} and \citet{Naor1991} to prove that graph coloring in cycles requires $\Omega(\log^* n)$ rounds. This was more recently used to prove lower bounds for \emph{sinkless orientations}, algorithmic Lov\'asz local lemma, and $\Delta$-coloring~\cite{Brandt2016,chang18complexity}, as well as to prove lower bounds for \emph{weak coloring}~\cite{Brandt2019automatic,BalliuHOS19}.

In principle, the approach can be used with \emph{any} locally checkable graph problem, in a mechanical manner~\cite{Brandt2019automatic}. However, if one starts with a natural problem $\Pi$ (e.g.\ MIS, MM, or graph coloring) and applies the speedup simulation in a mechanical manner, the end result is typically a problem $\Pi'$ that does not have any natural interpretation or simple description, and it gets quickly exponentially worse. The key technical challenge that the present work overcomes is the construction of a sequence of nontrivial problems $\Pi_1, \Pi_2, \dotsc$ such that each of them has a relatively simple description and we can nevertheless apply speedup simulation for any consecutive pair of them.
Note that the simplicity of the descriptions is crucial in two ways: it allows us to prove the desired relation between the complexities of any two subsequent problems in the sequence, and it enables us to show that the problems indeed cannot be solved in $0$ rounds. 

The formalism that we use is closely related to~\cite{Brandt2019automatic}---in essence, we generalize the formalism from graphs to hypergraphs, then represent the hypergraph as a bipartite graph, and we arrive at the formalism that we use in the present work to study maximal matchings in bipartite graphs.

\section{Lower bound in the port-numbering model}\label{sec:deterministic}

Consider the following setting: We have a \emph{$\Delta$-regular bipartite graph}; the nodes in one part are white and in the other part black. Each node has a \emph{port numbering} for the incident edges; the endpoints of the edges incident to a node are numbered in some arbitrary order with $1,2,\dotsc,\Delta$. See \figureref{fig:network} for an illustration.

\begin{figure}
    \centering
    \includegraphics[page=1,scale=0.8]{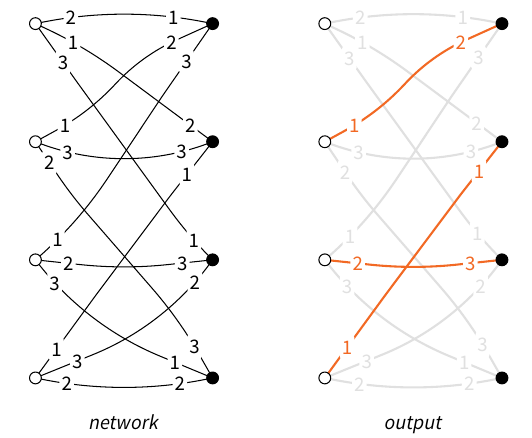}
    \Description{}
    \caption{A $3$-regular bipartite port-numbered network and a maximal matching.}\label{fig:network}
\end{figure}

The graph represents the topology of a communication network: each node is a computer, and each edge is a communication link. Initially each computer only knows its own color (black or white) and the number of ports ($\Delta$); the computers are otherwise identical. Computation proceeds in \emph{synchronous communication rounds}---in each round, each node can send an arbitrary message to each of its neighbors, then receive a message from each of its neighbors, and update its own state. After some $T$ communication rounds, all nodes have to stop and announce their own part of the solution; here $T$ is the \emph{running time} of the algorithm.

We are interested in algorithms for finding a \emph{maximal matching}; eventually each node has to know whether it is matched and in which port. There is a very simple algorithm that solves this in $T = O(\Delta)$ rounds~\cite{Hanckowiak1998}: In iteration $i = 1,2,\dotsc,\Delta$, unmatched white nodes send a proposal to their port number $i$, and black nodes accept the first proposal that they receive, breaking ties using port numbers. See \figureref{fig:proposal} for an example.

\begin{figure}
    \centering
    \includegraphics[page=2,scale=0.8]{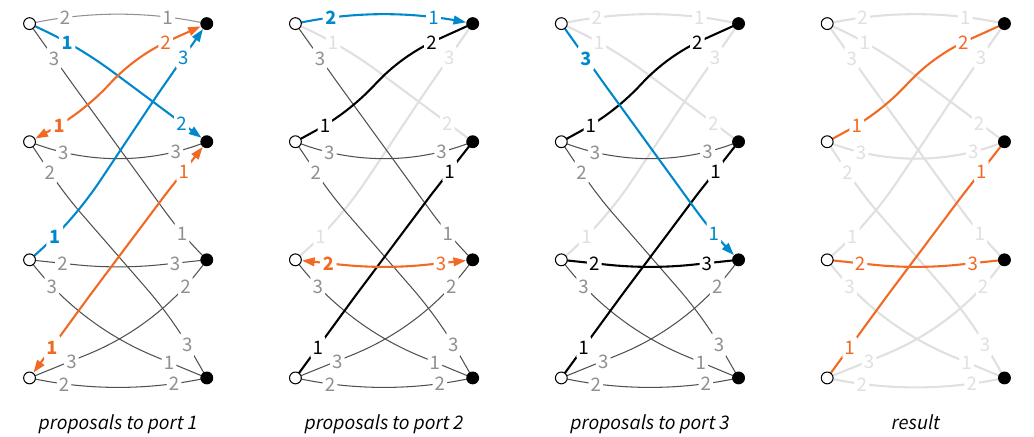}
    \Description{}
    \caption{The proposal algorithm finds a maximal matching in $O(\Delta)$ rounds---orange arrows are accepted proposals and blue arrows are rejected proposals.}\label{fig:proposal}
\end{figure}

Hence bipartite maximal matching can be solved in $O(\Delta)$ rounds in $\Delta$-regular two-colored graphs, and the running time is independent of the number of nodes. Surprisingly, nobody has been able to tell if this algorithm is optimal, or anywhere close to optimal. There are no algorithms that break the linear-in-$\Delta$ barrier (without introducing some dependency on $n$ in the running time), and there are no nontrivial lower bounds---we have not been able to exclude even the possibility of solving maximal matchings in this setting in e.g.\ $10$ rounds, independently of $\Delta$ and $n$. If we look at a bit more general setting of graphs of degree at most $\Delta$ (instead of $\Delta$-regular graphs), there is a lower bound of $\Omega(\log \Delta / \log \log \Delta)$~\cite{Kuhn2004,Kuhn2006,kuhn16local}, but there is still an exponential gap between the upper and the lower bound.

In this section we show that the trivial proposal algorithm is indeed optimal: there is no algorithm that finds a maximal matching in $o(\Delta)$ rounds in this setting. We will later extend the result to more interesting models of computing, but for now we will stick to the case of deterministic algorithms in the port-numbering model, as it is sufficient to explain all key ideas.

\subsection{Lower bound idea}

Our plan is to prove a lower bound using a \emph{speedup simulation argument}~\cite{Linial1987,Linial1992,Naor1991,Brandt2016,chang18complexity,Brandt2019automatic,Balliu2019}. The idea is to define a sequence of graph problems $\Pi_1, \Pi_2, \dotsc, \Pi_k$ such that if we have an algorithm $A_i$ that solves $\Pi_i$ in $T_i$ rounds, we can construct an algorithm $A_{i+1}$ that solves $\Pi_{i+1}$ strictly faster, in $T_{i+1} \le T_i - 1$ rounds (unless $T_i = 0$, in which case we would get that $T_{i+1}=0$). Put otherwise, we show that solving $\Pi_i$ takes at least one round more than solving $\Pi_{i+1}$. Then if we can additionally show that $\Pi_k$ is still a nontrivial problem that cannot be solved in zero rounds, we know that the complexity of $\Pi_1$ is at least $k$ rounds. 

Now we would like to let $\Pi_1$ be the maximal matching problem, and identify a suitable sequence of relaxations of the maximal matching problem. A promising candidate might be, e.g., the following problem that we call here a \emph{$k$-matching} for brevity.
\begin{definition}[$k$-matching]\label{def:k-matching}
    Given a graph $G = (V,E)$, a set of edges $M \subseteq E$ is a $k$-matching if
    \begin{enumerate}
        \item every node is incident to at most $k$ edges of $M$,
        \item if a node is not incident to any edge of $M$, then all of its neighbors are incident to at least one edge of $M$.
    \end{enumerate}
\end{definition}
Note that with this definition, a $1$-matching is exactly the same thing as a maximal matching. Also it seems that finding a $k$-matching is easier for larger values of $k$. For example, we could modify the proposal algorithm so that in each iteration white nodes send $k$ proposals in parallel, and stop as soon as at least one is accepted, and this way find a $k$-matching in $O(\Delta/k)$ rounds.

We could try to define $\Pi_i$ as the problem of finding an $i$-matching. Then, we could try to prove that, given an algorithm for finding an $i$-matching, we can construct a strictly faster algorithm for finding an $(i+1)$-matching. 
Unfortunately, a direct attack along these lines does not seem to work, but this serves nevertheless as a useful guidance that will point us in the right direction.

\subsection{Formalism and notation}

We will first make the setting as simple and localized as possible. It will be convenient to study graph problems that are of the following form---we call these \emph{edge labeling problems}:
\begin{enumerate}
    \item The task is to label the edges of the bipartite graph with symbols from some \emph{alphabet} $\Sigma$.
    \item A \emph{problem specification} is a pair $\Pi = (W,B)$, where $W$ is the set of feasible labelings of the edges incident to a white node, and $B$ is the set of feasible labelings for the edges incident to a black node.
\end{enumerate}

In these problems, the feasibility of a solution does not depend on the port numbering. Hence each member of $W$ and $B$ is a \emph{multiset} that contains $\Delta$ elements from alphabet $\Sigma$. For example, if we have $\Sigma = \{0,1\}$ and $\Delta = 3$, then $W = \bigl\{ \{0,0,0\}, \{0,0,1\} \bigr\}$ indicates that a white node is happy (that is, the labeling of its incident edges is correct according to $W$) if it is incident to exactly $0$ or $1$ edges with label~$1$.

However, for brevity we will here represent multisets as \emph{words}, and write e.g.\ $W = \bigl\{ 0\s0\s0, 0\s0\s1 \bigr\}$. We emphasize that the order of the elements does not matter here, and we could equally well write e.g.\ $W = \bigl\{ 0\s0\s0, 0\s1\s0 \bigr\}$. Now that $W$ and $B$ are languages over alphabet $\Sigma$, we can conveniently use regular expressions to represent them. When $x_1, x_2, \dotsc, x_k \in \Sigma$ are symbols of the alphabet, we use the shorthand notation $[x_1 x_2 \dotsc x_k] = ( x_1 | x_2 | \dotso | x_k )$. With this notation, we can represent the above example concisely as $W = 0\s0\s0 \mid 0\s0\s1$, or $W = 0\s0\s[01]$, or even $W = 0^2\s[01]$.

\subsubsection{Example: encoding maximal matchings}

\begin{figure}
    \centering
    \includegraphics[page=3,scale=0.8]{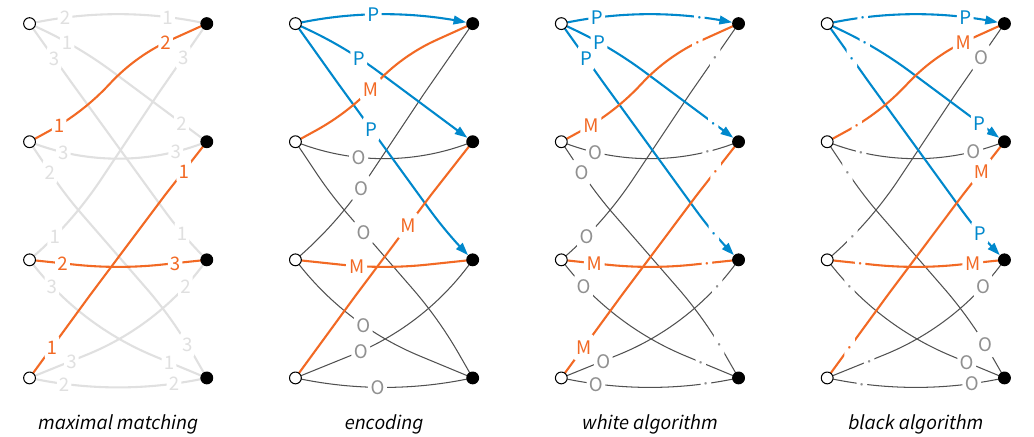}
    \Description{}
    \caption{Encoding maximal matchings with $\Sigma = \{ \M, \P, \O \}$.}\label{fig:mm-mpo}
\end{figure}

The most natural way to encode maximal matchings would be to use e.g.\ labels $0$ and $1$ on the edges, with $1$ to indicate an edge in the matching. However, this is not compatible with the above formalism: we would have to have $0^\Delta \in W$ and $0^\Delta \in B$ to allow for unmatched nodes, but then we would also permit a trivial all-$0$ solution. To correctly capture the notion of maximality, we will use three labels, $\Sigma = \{ \M, \P, \O \}$, with the following rules:
\begin{equation}
\begin{aligned}
W &= \M\s\O^{\Delta-1} \bigm| \P^{\Delta}, \\
B &= \M\s[\P\O]^{\Delta-1} \bigm| \O^{\Delta}.
\end{aligned}
\label{eq:mm}
\end{equation}
For a matched white node, one edge is labeled with an $\M$ (matched) and all other edges are labeled with an $\O$ (other). However, for an unmatched white node, all incident edges have to be labeled with a $\P$ (pointer); the intuition is that $\P$ points to a matched black neighbor. The rules for the black nodes ensure that pointers do not point to unmatched black nodes (a $\P$ implies exactly one~$\M$), and that black nodes are unmatched only if all white neighbors are matched (all incident edges labeled with $\O$s). See \figureref{fig:mm-mpo} for an illustration.

\subsubsection{White and black algorithms}

Let $\Pi = (W,B)$ be an edge labeling problem. We say that $A$ is a \emph{white algorithm} that solves $\Pi$ if in $A$ each white node outputs a labeling of its incident edges, and such a labeling forms a feasible solution to $\Pi$. Black nodes produce an empty output.

Conversely, in a \emph{black algorithm}, each black node outputs the labels of its incident edges, and white nodes produce an empty output. See \figureref{fig:mm-mpo} for illustrations. Note that a black algorithm can be easily turned into a white algorithm if we use one additional communication round, and vice versa.

\subsubsection{Infinite trees vs.\ finite regular graphs}\label{sec:treesvsgraphs}

It will be convenient to first present the proof for the case of infinite $\Delta$-regular trees. In essence, we will show that any algorithm $A$ that finds a maximal matching in $T = o(\Delta)$ rounds will fail around some node $u$ in some infinite $\Delta$-regular tree $G$ (for some specific port numbering), where failing around node $u$ means that the labels assigned to the edges incident to $u$ form a configuration that is not allowed by the problem. Then it is also easy to construct a finite $\Delta$-regular graph $G'$ such that the radius-$(T+1)$ neighborhood of $u$ in $G$ (including the port numbering) is isomorphic to the radius-$(T+1)$ neighborhood of some node $u'$ in $G'$, and therefore $A$ will also fail around $u'$ in $G'$.

\subsection{Parametrized problem family}

We will now introduce a parametrized family of problems $\Pi_\Delta(x,y)$, where $x+y \le \Delta$. The problem is defined so that $\Pi_\Delta(0,0)$ is equivalent to maximal matchings \eqref{eq:mm} and the problem becomes easier when we increase $x$ or $y$. We will use the alphabet $\Sigma = \{ \M, \P, \O, \X \}$, where $\M$, $\P$, and $\O$ have a role similar to maximal matchings and $\X$ acts as a \emph{wildcard}. We define $\Pi_\Delta(x,y) = \bigl(W_\Delta(x,y),\allowbreak B_\Delta(x,y)\bigr)$, where
\begin{equation}
\begin{aligned}
W_\Delta(x,y) &= \Bigl( \M\s\O^{d-1} \Bigm| \P^d \Bigr) \s \O^y\s\X^x, \\
B_\Delta(x,y) &= \Bigl( [\M\X]\s[\P\O\X]^{d-1} \Bigm| [\O\X]^d \Bigr) \s[\P\O\X]^y \s[\M\P\O\X]^x,
\end{aligned}
\label{eq:pixyd}
\end{equation}
where $d = \Delta-x-y$.

The following partial order represents the ``strength'' of the symbols from the perspective of black nodes:
\begin{equation}
\begin{tikzcd}[column sep=small, row sep=tiny]
& \M \arrow[rd] & \\
& & \X \\
\P \arrow[r] & \O \arrow[ru] & \\
\end{tikzcd}
\label{eq:4states}
\end{equation}
The interpretation is that from the perspective of $B_\Delta(x,y)$, symbol $\X$ is feasible wherever $\M$ or $\O$ is feasible, and $\O$ is feasible wherever $\P$ is feasible. Hence, we say that $\X$ is stronger than $\M$, $\O$, and $\P$, and that $\O$ is stronger than $\P$. Furthermore, all relations are strict in the sense that e.g.\ replacing an $\X$ with an $\M$ may lead to a word not in $B_\Delta(x,y)$.

Here are three examples of problems in family $\Pi_\Delta(\cdot,\cdot)$, with some intuition (from the perspective of a white algorithm):
\begin{itemize}
    \item $\Pi_\Delta(0,0)$: Maximal matching. Note that we cannot use symbol $\X$ at all, as it does not appear in $W_\Delta(0,0)$.
    \item $\Pi_\Delta(0,1)$: Unmatched white nodes will use $\O$ instead of $\P$ once---note that by \eqref{eq:4states} this is always feasible for the black node at the other end of the edge. Unmatched black nodes can accept $\P$ instead of $\O$ once.
    \item $\Pi_\Delta(1,0)$: All white nodes will use $\X$ instead of $\P$ or $\O$ once---again, this is always feasible. All black nodes can accept anything from one port.
\end{itemize}
In essence, $\Pi_\Delta(0,y)$ resembles a problem in which we can violate maximality, while in $\Pi_\Delta(x,0)$ we can violate the packing constraints.

\subsection{Speedup simulation}\label{ssec:speedup}

\begin{figure}
    \centering
    \includegraphics[page=4,scale=0.9]{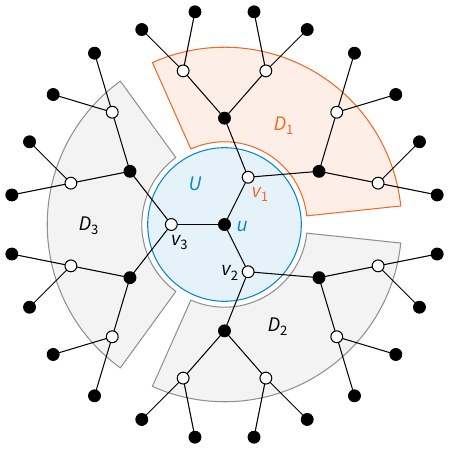}
    \Description{}
    \caption{Speedup simulation, for $T=2$ and $\Delta=3$. The radius-$(T-1)$ neighborhood of $u$ is $U$, and the radius-$T$ neighborhood of $v_i$ is $V_i = U \cup D_i$. The key observation is that $D_i$ and $D_j$ are disjoint for $i \ne j$.}\label{fig:speedup}
\end{figure}

Assume that $A$ is a white algorithm that solves $\Pi_\Delta(x,y)$ in $T \ge 1$ rounds in trees, for a sufficiently large $\Delta$. Throughout this section, let $d = \Delta - x - y$, and assume that $\Delta \ge 2x+y+1$.
We are going to proceed as follows.
First, we construct a black algorithm $A_1$ that runs in $T-1$ rounds and works as follows.
Given a $T-1$-radius neighborhood $U$, it extends $U$ into a $T$-radius neighborhood around a white node in all possible ways, then it simulates $A$ in all the obtained neighborhoods, and finally, for each incident edge, it outputs the set of the outputs given by $A$.
On each edge, the output of algorithm $A_1$ is a set, and there are $15$ possible sets that could appear (namely all non-empty subsets of the output label set for $\Pi_\Delta(x,y)$).
As our goal is to obtain a problem sequence where each problem is from the family $\Pi_\Delta(\cdot,\cdot)$, we need to reduce the number of labels back to $4$.
We start by reducing the number of possible outputs to $6$, by defining a new algorithm $A_2$ that first executes $A_1$ and then maps the aforementioned sets to labels, such that different sets could be mapped to the same label.
We could try to exactly characterize the problem solved by $A_2$, by listing the constraints defined over a set of $6$ labels that an output of $A_2$ satisfies. Instead, we are going to \emph{simplify} the output of $A_2$, and reduce the number of labels from $6$ to $4$ by \emph{identifying} some labels, which yields some algorithm $A_3$.
Finally, we are going to define algorithm $A_4$, that does nothing else than executing $A_3$ and then renaming the obtained labels, and we are going to show that $A_4$ solves $\Pi_\Delta(x+1,y+x)$, but in a graph where the role of black and white nodes is reversed. For a more intuitive explanation of this approach, please see Section ~\ref{sec:behind}.

\subsubsection{Algorithm $A_1$}

We will first construct a black algorithm $A_1$ that runs in time $T-1$, as follows:

\begin{framed}
    Each black node $u$ gathers its radius-$(T-1)$ neighborhood $U$; see \figureref{fig:speedup}. Let the white neighbors of $u$ be $v_1, v_2, \dotsc, v_\Delta$. Let $V_i$ be the radius-$T$ neighborhood of $v_i$, and let $D_i = V_i \setminus U$ be the part of $V_i$ that $u$ does not see.
    
    For each $i$, go through all possible inputs that one can assign to $D_i$; here the only unknown part is the port numbering that we have in the region $D_i$. Then simulate $A$ for each possible input and see how $A$ labels the edge $e_i = \{u,v_i\}$. Let $S_i$ be the set of labels that $A$ assigns to edge $e_i$ for some input $D_i$.
    
    Algorithm $A_1$ labels edge $e_i$ with set $S_i$.
\end{framed}

\subsubsection{Algorithm $A_2$}

Now since the output alphabet of $A$ is $\Sigma = \{ \M, \P, \O, \X \}$, the new output alphabet of $A_1$ consists of its 15 nonempty subsets. We construct another black algorithm $A_2$ with alphabet $\{\bX, \bOX,\allowbreak \bPOX,\allowbreak \bMX,\allowbreak \bMOX,\allowbreak \bMPOX\}$ that simulates $A_1$ and then maps the output of $A_1$ as follows (see Figure~\ref{fig:speedup-example} for an example):
\begin{align*}
\{\X\} &\mapsto \bX, \\
\{\M\}, \{\M,\X\} &\mapsto \bMX, \\
\{\O\}, \{\O,\X\} &\mapsto \bOX, \\
\{\M,\O\}, \{\M,\O,\X\} &\mapsto \bMOX, \\
\{\P\}, \{\P,\O\}, \{\P,\X\}, \{\P,\O,\X\} &\mapsto \bPOX, \\
\{\M,\P\}, \{\M,\P,\O\}, \{\M,\P,\X\}, \{\M,\P,\O,\X\} &\mapsto \bMPOX.
\end{align*}
Here the intuition is that we first make each set maximal w.r.t.\ \eqref{eq:4states}: for example, whenever we have a set with a $\P$, we also add an $\O$, and whenever we have a set with an $\O$, we also add an $\X$. This results in only six maximal sets, and then we replace e.g.\ the maximal set $\{\M,\O,\X\}$ with the label $\bMOX$.

\subsubsection{Output of $A_2$}

Let us analyze the output of $A_2$ for a black node. Fix a black node $u$ and its neighborhood~$U$. The key property is that regions $D_1, D_2, \dotsc$ in Figure~\ref{fig:speedup} do not overlap---hence if there is some input for $D_1$ that makes $v_1$ output some label $L_1$, and another input for $D_2$ that makes $v_2$ output some label $L_2$, we can also construct an input in which $v_1$ outputs $L_1$ and at the same time $v_2$ outputs $L_2$. We make the following observations:
\begin{enumerate}
    \item There can be at most $x+1$ edges incident to $u$ with a label in $\{\bMX,\bMOX,\bMPOX\}$. If there were $x+2$ such edges, say $e_1, e_2, \dotsc, e_{x+2}$, then it means we could fix $D_1$ such that $A$ outputs $\M$ for $e_1$, and simultaneously fix $D_2$ such that $A$ outputs $\M$ for $e_2$, etc. But this would violate the property that $A$ solves $\Pi_\Delta(x,y)$, as all words of $B_\Delta(x,y)$ contain at most $x+1$ copies of~$\M$.
    \item If there are at least $x+y+1$ edges with a label in $\{\bPOX,\bMPOX\}$, then there has to be at least one edge with a label in $\{\bX,\bMX\}$. Otherwise we could choose $D_i$ so that $A$ outputs $\P$ on $x+y+1$ edges, and there is no $\M$ or $\X$. But all words of $B_\Delta(x,y)$ with at least $x+y+1$ copies of $\P$ contain also at least one $\M$ or $\X$.
\end{enumerate}

\subsubsection{Algorithm $A_3$}

\begin{figure}
    \centering
    \begin{tabular}{lccccccc}
        \toprule
        Edge & $e_1$ & $e_2$ & $e_3$ & $e_4$ & $e_5$ & $e_6$ & $e_7$ \\
        \midrule
        Part of $B_7(1,1)$ & $[\M\X]$ & $[\P\O\X]$ & $[\P\O\X]$ & $[\P\O\X]$ & $[\P\O\X]$ & $[\P\O\X]$ & $[\M\P\O\X]$ \\
        \midrule
        White algorithm $A$ & $\M$ & $\X$ & $\O$ & $\X$ & $\P$ & $\O$ & $\M$ \\
        Black algorithm $A_1$ & $\{\M,\O\}$ & $\{\X\}$ & $\{\O\}$ & $\{\P,\X\}$ & $\{\P\}$ & $\{\P,\O\}$ & $\{\M,\P\}$ \\
        Black algorithm $A_2$ & $\bMOX$ & $\bX$ & $\bOX$ & $\bPOX$ & $\bPOX$ & $\bPOX$ & $\bMPOX$ \\
        Black algorithm $A_3$ & $\bMPOX$ & $\bMX$ & $\bPOX$ & $\bPOX$ & $\bPOX$ & $\bPOX$ & $\bMPOX$ \\
        Black algorithm $A_4$ & $\X$ & $\M$ & $\O$ & $\O$ & $\O$ & $\O$ & $\X$ \\
        \midrule
        Part of $W_7(2,2)$ & $\X$ & $\M$ & $\O$ & $\O$ & $\O$ & $\O$ & $\X$ \\
        \bottomrule
    \end{tabular}
    \Description{}
    \caption{Speedup simulation for $\Pi_7(1,1)$: an example of some possible outputs around a black node.}\label{fig:speedup-example}
\end{figure}

We construct yet another black algorithm $A_3$ that modifies the output of $A_2$ so that we replace labels only with larger labels according to the following partial order, which represents subset inclusion:
\begin{equation}
\begin{tikzcd}[column sep=small]
\bX \arrow[r]\arrow[d] & \bOX \arrow[r]\arrow[d] & \bPOX \arrow[d] \\
\bMX \arrow[r] & \bMOX \arrow[r] & \bMPOX
\end{tikzcd}
\label{eq:6diag}
\end{equation}
There are two cases:
\begin{enumerate}
    \item There are at most $x+y$ copies of $\bPOX$ on edges incident to $u$. We also know that there are at most $x+1$ copies of labels $\{\bMX,\bMOX,\bMPOX\}$. Hence there have to be at least $\Delta-{(x+y)}-{(x+1)} = d-x-1$ copies of labels $\{\bX,\bOX\}$. We proceed as follows:
    \begin{itemize}
        \item Replace all of $\{\bMX,\bMOX\}$ with $\bMPOX$.
        \item Replace some of $\{\bX,\bOX,\bPOX\}$ with $\bMPOX$ so that the total number of $\bMPOX$ is exactly $x+1$.
        \item Replace some of the remaining $\{\bX,\bOX\}$ with $\bPOX$ so that the total number of $\bPOX$ is exactly $x+y$.
        \item Replace all of the remaining $\bX$ with $\bOX$.
    \end{itemize}
    We are now left with exactly $x+1$ copies of $\bMPOX$, exactly $x+y$ copies of $\bPOX$, and exactly $d-x-1$ copies of $\bOX$.
    \item There are more than $x+y$ copies of $\bPOX$. Then we know that there is at least one copy of $\{\bX,\bMX\}$. We first proceed as follows:
    \begin{itemize}
        \item Replace all $\bOX$ with $\bPOX$ and all $\bMOX$ with $\bMPOX$.
        \item If needed, replace one $\bX$ with $\bMX$ so that there is at least one copy of $\bMX$.
        \item Replace all remaining copies of $\bX$ with $\bPOX$.
    \end{itemize}
    At this point we have at least one $\bMX$ and all other labels are in $\{\bMPOX,\bPOX\}$. Furthermore, as we originally had at most $x+1$ copies of $\{\bMX,\bMOX,\bMPOX\}$, and we created at most one new label in this set, we have now got at most $x+2$ copies of $\{\bMX,\bMPOX\}$. Hence we have got at least $d+y-2$ copies of $\bPOX$. We continue:
    \begin{itemize}
        \item Replace some of $\bMX$ with $\bMPOX$ so that the total number of $\bMX$ is exactly $1$.
        \item Replace some of $\bPOX$ with $\bMPOX$ so that the total number of $\bPOX$ is exactly $d+y-2$ and hence the total number of $\bMPOX$ is $x+1$.
    \end{itemize}
\end{enumerate}

Note that we have completely eliminated labels $\bX$ and $\bMOX$ and we are left with the following four labels:
\begin{equation}
\begin{tikzcd}[column sep=small]
\cdot \arrow[r]\arrow[d] & \bOX \arrow[r]\arrow[d] & \bPOX \arrow[d] \\
\bMX \arrow[r] & \cdot \arrow[r] & \bMPOX
\end{tikzcd}
\label{eq:64diag}
\end{equation}

\subsubsection{Algorithm $A_4$}

Finally, we construct an algorithm $A_4$ that maps the output of $A_3$ as follows:
\begin{equation}
\bMX \mapsto \M, \quad
\bOX \mapsto \P, \quad
\bPOX \mapsto \O, \quad
\bMPOX \mapsto \X.
\end{equation}
Note that after this mapping, poset \eqref{eq:64diag} restricted to $\{\M,\P,\O,\X\}$ is the same as poset \eqref{eq:4states}. 
One may have expected $\bOX$ to be renamed to $\O$ and $\bPOX$ to be renamed to $\P$, but we did the opposite. Let us discuss the intuition behind this choice by taking maximal matching as an example. In this problem, $\P$ is used by white notes (that are active) to certify that all black neighbors (that are passive) are matched. Similarly, we can see label $\O$ as a label that can be used by \emph{black} nodes to certify that they have white matched neighbors. Hence, we can see \emph{both} $\P$ and $\O$ as pointers. The renaming that we use allows us to see $\P$ as the pointer of the \emph{active} side and $\O$ as the pointer of the \emph{passive} side, even if the role of active and passive swapped.

\subsubsection{Output of $A_4$}

Let us analyze the output of $A_4$ for a white node $v$. The key observation is that the output of $A$ is contained in the sets that $A_1$ outputs. More precisely, algorithm $A_1$ outputs sets based on its $T-1$-radius neighborhood, and if there is an extension of this neighborhood such that algorithm $A$, executed on the extended neighborhood, outputs some label on the edge corresponding to this extension, then this output label is contained in the set given by $A_1$ on this edge. Since all the possible outputs of $A$ are in $( \M\s\O^{d-1} | \P^d ) \s \O^y\s\X^x$, there are two cases:
\begin{enumerate}
    \item We have a neighborhood in which the output of $A$ at $v$ matches $\M\s\O^{d-1+y}\s\X^x$. Then:
    \begin{itemize}
        \item There is at least one edge incident to $v$ such that the output of $A_1$ contains an $\M$. Therefore the output of $A_2$ on this edge is in $\{\bMX,\bMOX,\bMPOX\}$. As we followed \eqref{eq:6diag}, the output of $A_3$ for this edge is also in $\{\bMX,\bMOX,\bMPOX\}$. After remapping, the output of $A_4$ for this edge is in $\{\M,\X\}$.
        \item There are at least $d-1+y$ edges incident to $v$ such that the output of $A_1$ contains an $\O$. By a similar argument, the output of $A_3$ for each of these edges is in $\{\bOX,\bMOX,\bPOX,\bMPOX\}$, and the output of $A_4$ is hence in $\{\P,\O,\X\}$.
    \end{itemize}
    \item We have a neighborhood in which the output of $A$ at $v$ matches $\P^d\s\O^y\s\X^x$. Then by a similar reasoning:
    \begin{itemize}
        \item We have $d$ edges incident to $v$ such that the output of $A_3$ is in $\{\bPOX,\bMPOX\}$ and hence the output of $A_4$ is in $\{\O,\X\}$.
        \item We have $y$ edges incident to $v$ such that the output of $A_3$ is in $\{\bOX,\bMOX,\bPOX,\bMPOX\}$ and hence the output of $A_4$ is in $\{\P,\O,\X\}$.
    \end{itemize}
\end{enumerate}
Hence the neighborhood of a white node in $A_4$ satisfies
\[
\begin{split}
W^* ={} &[\M\X] \s [\P\O\X]^{d-1+y} \s [\M\P\O\X]^{x} \bigm| [\O\X]^d\s[\P\O\X]^y\s[\M\P\O\X]^{x} \\
={} &\Bigl( [\M\X] \s [\P\O\X]^{d-1} \bigm| [\O\X]^d \Bigr) \s[\P\O\X]^y\s[\M\P\O\X]^{x}.
\end{split}
\]

By construction, the output of a black node in $A_3$ follows one of these patterns:
\begin{enumerate}
    \item $x+1$ copies of $\bMPOX$, $x+y$ copies of $\bPOX$, and $d-x-1$ copies of $\bOX$.
    \item $1$ copy of $\bMX$, $d+y-2$ copies of $\bPOX$, and $x+1$ copies of $\bMPOX$.
\end{enumerate}
Hence after remapping, the output of a black node in $A_4$ satisfies
\[
\begin{split}
B^* &= \P^{d-x-1}\s\O^{x+y}\s\X^{x+1} \bigm| \M\s\O^{d+y-2}\s\X^{x+1} \\
&= \Bigl(\M\s\O^{d-x-2} \bigm| \P^{d-x-1}\bigr)\s\O^{x+y}\s\X^{x+1}.
\end{split}
\]
Hence $A_4$ solves $\Pi^* = (W^*,B^*)$. Now observe that $W^* = B_\Delta(x,y)$ and $B^* = W_\Delta(x+1,y+x)$.

Furthermore, observe that $B_\Delta(x,y) \subseteq B_\Delta(x+1,y+x)$. Hence $A_4$ also solves the problem
\[
\Pi'_\Delta(x+1,y+x) = \big(B_\Delta(x+1,y+x),\, W_\Delta(x+1,y+x)\bigr),
\]
which is otherwise exactly the same as $\Pi_\Delta(x+1,y+x)$, but the roles of the black and white nodes are reversed.

\subsubsection{Conclusion}

We started with the assumption that $A$ is a white algorithm that solves $\Pi_\Delta(x,y)$ in $T$ rounds. We constructed a black algorithm $A_1$ that runs in $T-1$ rounds. Then we applied three steps of local postprocessing to obtain algorithm $A_4$; note that $A_4$ is also a black algorithm with the running time $T-1$, as all postprocessing steps can be done without communication. We observed that $A_4$ solves $\Pi'_\Delta(x+1,y+x)$.

Now if there is a black algorithm that solves $\Pi'_\Delta(x+1,y+x)$ in $T-1$ rounds, we can reverse the roles of the colors and obtain a white algorithm that solves $\Pi_\Delta(x+1,y+x)$ in $T-1$ rounds. We have the following lemma:

\begin{lemma}\label{lem:speedup}
    Assume that there exists a white algorithm that solves $\Pi_\Delta(x,y)$ in $T \ge 1$ rounds in trees, for $\Delta \ge 2x+y+1$. Then there exists a white algorithm that solves $\Pi_\Delta(x+1,y+x)$ in $T-1$ rounds in trees.
\end{lemma}

By a repeated application of \lemmaref{lem:speedup}, we obtain the following corollary:

\begin{corollary}\label{cor:speedup}
    Assume that there exists a white algorithm that solves $\Pi_\Delta(x,y)$ in $T$ rounds in trees, for $\Delta \ge x+y+T(x+1+(T-1)/2)$. Then there is a white algorithm that solves $\Pi_\Delta(x',y')$ in $0$ rounds in trees for
    \begin{align*}
    x' &= x + T, \\
    y' &= y + T\bigl(x + (T-1)/2\bigr).
    \end{align*}
\end{corollary}

\subsection{Base case}

Now to make use of \corollaryref{cor:speedup}, it is sufficient to show that there are problems in the family $\Pi_\Delta(x,y)$ that cannot be solved with $0$-round white algorithms:

\begin{lemma}\label{lem:base}
    There is no white algorithm that solves $\Pi_\Delta(x,y)$ in $0$ rounds in trees when $\Delta \ge x+y+2$.
\end{lemma}
\begin{proof}
    Since by increasing $x$ and $y$ the problem does not get harder, it is enough to prove the statement for $\Delta = x+y+2$.
    As we are looking at deterministic algorithms in the port-numbering model, in a $0$-round white algorithm all white nodes produce the same output. By definition, the output has to be a permutation of some word $w \in W_\Delta(x,y)$. Now there are only two possibilities:
    \begin{enumerate}
        \item $w = \M\s\O^{y+1}\s\X^x$: There is some port $p$ that is labeled with $\M$. Now consider a black node $u$ such that for all white neighbors $v$ of $u$, port $p$ of $v$ is connected to $u$. Then all $\Delta = x+y+2$ edges incident to $u$ are labeled with an $\M$. However, by definition of $B_\Delta(x,y)$, at most $x+1$ edges incident to a black node can be labeled with an~$\M$.
        \item $w = \P^2\s\O^y\s\X^x$: By a similar reasoning, we can find a black node $u$ such that all $\Delta = x+y+2$ edges incident to $u$ are labeled with a $\P$. However, at most $x+y+1$ edges incident to a black node can be labeled with a~$\P$. \qedhere
    \end{enumerate}
\end{proof}

\subsection{Amplifying through \texorpdfstring{\boldmath $k$}{k}-matchings}

A straightforward application of \corollaryref{cor:speedup} and \lemmaref{lem:base} already gives a new lower bound---it just is not tight yet:

\begin{theorem}\label{thm:simple}
    Any deterministic algorithm that finds a maximal matching in the port-numbering model in $2$-colored graphs requires $\Omega(\sqrt{\Delta})$ rounds.
\end{theorem}
\begin{proof}
    Any algorithm that runs in $T$ rounds implies a white algorithm that runs in $O(T)$ rounds, and naturally a general algorithm also has to solve the problem correctly in trees.
    Let us assume for a contradiction that for any constant $c > 0$ and any positive integer $\Delta_0$ there exists $\Delta > \Delta_0$ such that there exists a white deterministic algorithm that finds a maximal matching in the port numbering model in $2$-colored graphs in less than $c \sqrt{\Delta}$ rounds.
    
     Recall that $\Pi_\Delta(0,0)$ is equal to the maximal matching problem. If we could solve it with a white algorithm in less than $c \sqrt{\Delta}$ rounds, then by \corollaryref{cor:speedup} we could also solve $\Pi_\Delta(x',y')$ in $0$ rounds, where $x' = c \sqrt{\Delta}$ and $y' = c \sqrt{\Delta} (c \sqrt{\Delta}-1)/2$. For small enough $c$ and large enough $\Delta_0$, $x'+y'+2 = c \sqrt{\Delta} + c \sqrt{\Delta} (c \sqrt{\Delta}-1)/2 \le \Delta$, which contradicts \lemmaref{lem:base}.
\end{proof}

However, we can easily amplify this and get a linear-in-$\Delta$ lower bound. Recall $k$-matchings from \defref{def:k-matching}. We make the following observation:

\begin{lemma}\label{lem:k-matching-pi}
    If there is an algorithm that finds a $k$-matching in $T$ rounds, there is also a white algorithm that solves $\Pi_\Delta(k-1,0)$ in $T+O(1)$ rounds.
\end{lemma}
\begin{proof}
    Consider a $k$-matching $M$, and consider a white node~$u$.
    \begin{enumerate}
        \item If $u$ is incident to at least one edge of $M$, then pick arbitrarily one $e \in M$ incident to $u$ and label $e$ with an $\M$. Label all other edges $\{u,v\} \in M \setminus \{e\}$ with an $\X$; note that this results in at most $k-1$ incident edges with an $\X$. Then label all other edges $\{u,v\} \notin M$ with $\O$ or $\X$ such that the total number of incident edges with an $\X$ is exactly $k-1$.
        \item Otherwise $u$ is not incident to any edge of $M$. In this case pick arbitrarily $k-1$ incident edges, label them with an $\X$, and label all other incident edges with a $\P$.
    \end{enumerate}
    At this point by construction white nodes satisfy $W_\Delta(k-1,0)$. It remains to be checked that black nodes satisfy $B_\Delta(k-1,0)$; to this end, consider a black node $v$. There are two cases:
    \begin{enumerate}
        \item Node $v$ is incident to at least one edge of $M$: All edges in $M$ are labeled with $\M$ or $\X$ and all edges not in $M$ are labeled with $\P$, $\O$, or $\X$. We have got at least one edge incident to $v$ labeled with $\M$ or $\X$, at most $k-1$ additional incident edges labeled with $\M$ or $\X$, and everything else in $[\P\O\X]$, which forms a word in \[[\M\X]\s[\P\O\X]^{\Delta-k} \s[\M\P\O\X]^{k-1}.\]
        \item Node $v$ is not incident to an edge of $M$: Then by definition all white neighbors $u$ of $v$ were incident to at least one edge of $M$, and hence all incident edges are labeled with $\O$ or $\X$, which forms a word~in
        \[[\O\X]^{\Delta-k+1} \s[\M\P\O\X]^{k-1}.\qedhere\]
    \end{enumerate}
\end{proof}

Now the same reasoning as \theoremref{thm:simple} gives the following result:

\begin{theorem}\label{thm:k-matching}
    Any deterministic algorithm that finds an $O(\sqrt{\Delta})$-matching in the port-numbering model in $2$-colored graphs requires $\Omega(\sqrt{\Delta})$ rounds.
\end{theorem}
\begin{proof}
	For a contradiction, assume that the theorem statement is not true, which, by Lemma~\ref{lem:k-matching-pi}, implies that there is a $T$-round algorithm that solves $\Pi_\Delta(x,y)$, for $x = O(\sqrt{\Delta})$, $y = 0$, and $T = o(\sqrt{\Delta})$. By \corollaryref{cor:speedup}, this implies that there is a $0$-round algorithm for $\Pi_\Delta(x',y')$, where $x' = O(\sqrt{\Delta})$ and $y' = o(\Delta)$. Note that, for $\Delta$ large enough, $\Delta - x' - y'$ is much larger than $2$. But by \lemmaref{lem:base}, since  $\Delta \ge x' + y' + 2$, $\Pi_\Delta(x',y')$ cannot be solved in $0$ rounds, and hence we reached a contradiction.
\end{proof}

Now it turns out that \theoremref{thm:k-matching} is tight, and we can use it to prove a tight lower bound for maximal matchings:

\begin{theorem}\label{thm:full}
    Any deterministic algorithm that finds a maximal matching in the port-numbering model in $2$-colored graphs requires $\Omega(\Delta)$ rounds.
\end{theorem}
\begin{proof}
    Assume that $A$ is an algorithm that finds a maximal matching in $f(\Delta) = o(\Delta)$ rounds. We use $A$ to construct algorithm $A'$ that finds an $O(\sqrt{\Delta})$-matching in $o(\sqrt{\Delta})$ rounds; then we can apply \theoremref{thm:k-matching} to see that $A$ cannot exist.
    
    \begin{figure}
        \centering
        \includegraphics[page=5,scale=0.9]{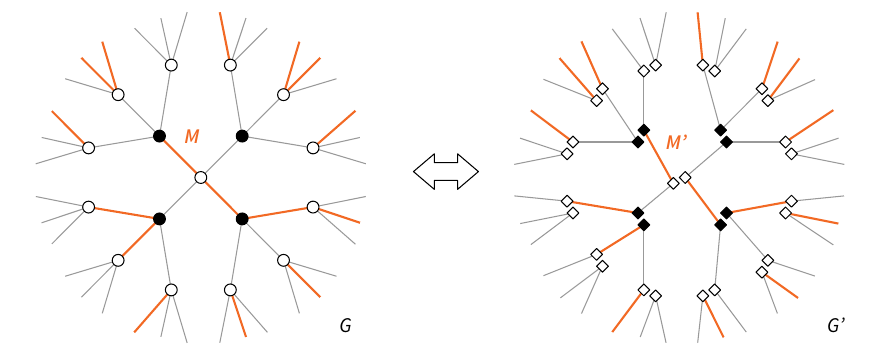}
        \Description{}
        \caption{Splitting nodes of degree $\Delta=4$ in $\sqrt{\Delta} = 2$ mininodes, each of degree $\sqrt{\Delta} = 2$. A maximal matching $M'$ of $G'$ defines a $2$-matching $M$ in $G$.}\label{fig:split}
    \end{figure}
    
    Algorithm $A'$ constructs a new virtual graph $G'$ by splitting each node $u$ of $G$ arbitrarily in $O(\sqrt{\Delta})$ \emph{mininodes} $u_1, u_2,\dotsc$, each of degree $O(\sqrt{\Delta})$; see \figureref{fig:split}. Then $A'$ simulates $A$ in graph $G'$ to find a maximal matching $M'$ of $G'$. Now $M'$ defines an $O(\sqrt{\Delta})$-matching $M$ in the original graph $G$. To see this, note that each mininode is incident to at most one edge of $M'$, and hence each original node is incident to $O(\sqrt{\Delta})$ edges of $M$. Furthermore, if a node $u$ is not incident to any edge of $M$, then all neighbors of each mininode are matched in $M'$, and hence all original neighbors of $u$ are incident to at least one edge of $M$.
    
    To conclude the proof, note that the maximum degree of the virtual graph $G'$ is $\Delta' = O(\sqrt{\Delta})$, and hence the simulation of $A$ in $G'$ completes in $f(\Delta') = o(\sqrt{\Delta})$ rounds.
\end{proof}

\subsection{Bounds for finite graphs}

We now prove the existence of a certain family of graphs, that will be later used to apply the ideas explained in Section~\ref{sec:treesvsgraphs}.
\begin{lemma}\label{lem:regulargraphs}
    There exists a constant $c>0$ such that for any $n$ large enough and any $\Delta$ satisfying that at least one of $n$ and $\Delta$ is even and $2 \le \Delta < n/10$, there exists a $\Delta$-regular connected graph $G = (V,E)$ where $|V|=n$  that contains a node $v \in V$ such that the radius-$t$ neighborhood of $v$ is isomorphic to a $\Delta$-regular tree, for some $t \ge c \log_\Delta n$.
\end{lemma}
\begin{proof}
    For $\Delta=2$ the statement follows by considering a cycle of $n$ nodes, so in the following we assume $\Delta>2$. A $\Delta$-regular balanced tree of depth $t$ contains $f_\Delta(t)=(\Delta(\Delta-1)^t-2)/(\Delta-2)$ nodes. We choose the maximal $t$ such that $n-f_\Delta(t)-(\Delta-1)^{t}>\Delta$. Note that we always have $t\ge1$, as $n>10\Delta$ and $\Delta+f_\Delta(1)+(\Delta-1)^1=3\Delta$. We build a $\Delta$-regular balanced tree $T$ of depth $t$. Tree $T$ has $\Delta(\Delta-1)^{t-1}$ leaves of degree 1, and all the other nodes are of degree $\Delta$. Finally we turn $T$ into a $\Delta$-regular graph: we add $(\Delta-1)^{t}$ new nodes, called \emph{bottom nodes}, and connect those nodes with the leaves; each bottom node is connected to $\Delta$ distinct leaf nodes, and each leaf node is connected to distinct $\Delta-1$ bottom nodes.

    So far we have got $f_\Delta(t)+(\Delta-1)^{t}$ nodes in graph $T$; we are still missing $N = n-f_\Delta(t)-(\Delta-1)^{t} > \Delta$ nodes. Notice that if $\Delta$ is odd, $f_\Delta(t)$ and $(\Delta-1)^{t}$ are even, meaning that $N$ is even. We now use the fact that for any $\Delta \ge 2$ and $N > \Delta$ such that one of them is even, a $\Delta$-regular connected graph $H$ of size $N$ always exists: arrange the $N$ nodes in a cycle; if $\Delta$ is even, node $i$ is connected to its $\Delta/2$ predecessors and successors; if $\Delta$ is odd, $N$ is even and you can also connect $i$ to the node diametrically opposed.

    So far we have got two connected components, $T$ and $H$, both of them $\Delta$-regular graphs with $n$ nodes in total. In order to connect $T$ and $H$, we pick an arbitrary edge $\{u,v\}$ of $H$ and an arbitrary edge $\{w,x\}$ adjacent to a bottom node $x$ of $T$. We remove such edges and add the edges $\{u,w\}$ and $\{v,x\}$.

    The obtained graph $G$ contains exactly $n$ nodes, it is $\Delta$-regular, and the $t$-radius ball around the node corresponding to the root of $T$ is a $\Delta$-regular tree. Moreover, we can notice that $t = \Omega(\log_\Delta f_\Delta(t))$. By the choice of $t$, we have $f_\Delta(t)-(\Delta-1)^{t}< n-\Delta < f_\Delta(t+1)-(\Delta-1)^{t+1}$, hence $t= \Omega(\log_\Delta n)$.
\end{proof}

\begin{lemma}\label{cor:full}
    There exists a constant $c>0$ such that for any $n$ and $\Delta$ large enough satisfying that at least one of $n$ and $\Delta$ is even and $\Delta < n/10$,
    any deterministic algorithm that finds a maximal matching in $\Delta$-regular $2$-colored graphs of $n$ nodes requires at least $c \min\{\Delta,\log_\Delta n\}$ rounds in the port-numbering model.
\end{lemma}
\begin{proof}
    We can apply the idea explained in Section~\ref{sec:treesvsgraphs} to convert the lower bound of \theoremref{thm:full}, presented for infinite $\Delta$-regular trees, into a lower bound for $\Delta$-regular graphs containing $n$ nodes. 
    
    Let us assume for a contradiction that for any constant $c > 0$ and any positive integers $n_0$ and $\Delta_0$ there exist $n > n_0$ and $\Delta > \Delta_0$ satisfying $\Delta < n/10$ such that there is a deterministic algorithm that finds a maximal matching in $\Delta$-regular $2$-colored graphs of $n$ nodes in less than $c \min\{\Delta,\log_\Delta n\}$ rounds.

    We fix $c$ to be the minimum of $0.1$, the constant $c$ of Lemma~\ref{lem:regulargraphs}, and the constant hidden in the $\Omega$-notation of \theoremref{thm:full}. We will later fix $n_0$ and $\Delta_0$ to be large enough constants.
    
    We run this algorithm on an infinite $\Delta$-regular tree, by claiming that the number of nodes is $n$. Since the running time of the algorithm is less than $c \log_\Delta n$, the algorithm sees at most $\Delta^{c \log_\Delta n} = n^c < n$ nodes, and hence does not detect that we lied about the size of the graph. Since $T $ is also less than $c \Delta$, by \theoremref{thm:full} the algorithm must fail in some neighborhood, if $\Delta$ is large enough. We fix $\Delta_0$ to make $\Delta$ satisfy such a requirement.
    
    Now, we fix $n_0$ large enough and apply Lemma~\ref{lem:regulargraphs} to construct a $\Delta$-regular graph of $n$ nodes containing a node with the same neighborhood in which the algorithm failed, and the algorithm must fail in such a graph as well, a contradiction.
\end{proof}

Note that $\Delta$ and $\log_\Delta n$ become roughly the same by setting $\Delta \approx \log n / \log \log n$. Hence, if we consider $\Delta$ to be an upper bound for the maximum degree, we directly get the following.

\begin{corollary}\label{cor:pn-maxdeg}
    Any deterministic algorithm that finds a maximal matching in $2$-colored graphs of maximum degree at most $\Delta$ requires $\Omega(\min\{\Delta,\log n / \log \log n\})$ rounds in the port-numbering model.
\end{corollary}

The following corollary gives a lower bound for the running time of any algorithm for maximal matching, if we express it solely as a function of $n$.

\begin{corollary}\label{cor:pn-maxn}
    Any deterministic algorithm that finds a maximal matching in $\Delta$-regular $2$-colored graphs requires $\Omega(\log n / \log \log n)$ rounds in the port-numbering model.
\end{corollary}

We have now come to the conclusion of this section: we have a tight linear-in-$\Delta$ lower bound for deterministic algorithms in the port-numbering model. In \sectionref{sec:randomized} we show how to extend the same argument so that it also covers randomized algorithms in the usual LOCAL model (and as a simple corollary, also gives a lower bound for MIS). The extension is primarily a matter of technicalities---we will use the same ideas and the same formalism as what we have already used in this section, and then analyze how the probability of a failure increases when we speed up randomized algorithms.

Before presenting the randomized lower bound, we briefly discuss the question of \emph{how} we came up with the right problem family $\Pi_\Delta(x,y)$ that enabled us to complete the lower bound proof---we hope similar ideas might eventually lead to a lower bound for other graph problems besides MM and MIS.

\subsection{Behind the scenes: how the right problem family was discovered}\label{sec:behind}

Let us now look back at the ingredients of the lower-bound proof. The most interesting part is the speedup simulation of \sectionref{ssec:speedup}. Here algorithm $A_1$ follows the standard idea that is similar to what was used already by \citet{Linial1987,Linial1992}. Algorithm $A_2$ is related to the idea of \emph{simplification by maximality} in~\cite{Brandt2019automatic}---this is something we can do without losing any power in the speedup simulation argument. On the other hand, algorithm $A_4$ merely renames the output labels. Hence algorithm $A_3$ is the only place in which truly novel ideas are needed.

Now if we already had the insight that the problem family $\Pi_\Delta(x,y)$ defined in \eqref{eq:pixyd} is suitable for the purpose of the speedup simulation, then inventing $A_3$ is not that hard with a bit of trial and error and pattern matching. However, at least to us it was not at all obvious that \eqref{eq:pixyd} would be the right relaxation of the maximal matching problem \eqref{eq:mm}.

While some amount of lucky guessing was needed, it was greatly simplified by following this approach:
\begin{itemize}
    \item We start with the formulation of \eqref{eq:mm}; let us call this problem $\Pi_0$. This is an edge labeling problem with $3$ labels.
    \item We apply the automatic speedup simulation framework of~\cite{Brandt2019automatic} to obtain a problem $\Pi_1$ that is \emph{exactly} one round faster to solve than $\Pi_0$. Then simplify $\Pi_1$ as much as possible without losing this property. It turns out that $\Pi_1$ is an edge labeling problem with $4$ labels.
    \item Repeat the same process to obtain problem $\Pi_2$, which is an edge labeling problem with $6$ labels.
    \item At this point we can see that the structure of problems $\Pi_0$, $\Pi_1$, and $\Pi_2$ is vaguely similar, but the set of labels is rapidly expanding and this also makes the problem description much more difficult to comprehend.
    \item A key idea is this: given problem $\Pi_2$ with $6$ labels, we can construct another problem $\Pi'_2$ that is at least as easy to solve as $\Pi_2$ by \emph{identifying some labels} of $\Pi_2$. For example, if the set of output labels in $\Pi_2$ is $\{1,2,3,4,5,6\}$, we could replace both $1$ and $2$ with a new label $12$ and obtain a new problem $\Pi'_2$ with the alphabet $\{12,3,4,5,6\}$. Trivially, given an algorithm for solving $\Pi_2$ we can also solve $\Pi'_2$ by remapping the outputs. However, it is not at all obvious that $\Pi'_2$ is a nontrivial problem.
    \item Many such identifications result in a problem $\Pi'_2$ that can be shown to be trivial (e.g.\ we can construct a problem that solves it in $0$ or $1$ round). However, by greedily exploring possible identifications we can find a way to map $6$ output labels to $4$ output labels such that the resulting problem $\Pi'_2$ still seems to be almost as hard to solve as the maximal matching problem.
    \item At this point we have obtained the following problems: $\Pi_0$ with $3$ labels, $\Pi_1$ with $4$ labels, and $\Pi'_2$ also with $4$ labels. For $\Pi_0$ we had labels $\{\M,\P,\O\}$ with a simple interpretation. The next step is to try to rename the labels of $\Pi_1$ and $\Pi'_2$ so that they would also use the familiar labels $\{\M,\P,\O\}$ plus some additional extra label $\X$. It turns out that there is a way to rename the labels so that both $\Pi_1$ and $\Pi'_2$ have some vague resemblance to the original formulation of~\eqref{eq:mm}. Here a helpful guidance was to consider posets similar to \eqref{eq:4states} and \eqref{eq:6diag} that visualize some parts of the problem structure, and try to find a labeling that preserves the structure of the poset.
\end{itemize}
Now, in essence, \eqref{eq:pixyd} is inspired by a relaxation and a generalization of problems $\Pi_0$, $\Pi_1$, and $\Pi'_2$ constructed above, and algorithm $A_3$ captures the key idea of mapping $6$ output labels to $4$ output labels.

While the idea is arguably vague and difficult to generalize, we highlight some ingredients that turned out to be instrumental:
\begin{itemize}
    \item We do not try to first guess a family of problems; we apply the speedup simulation framework from~\cite{Brandt2019automatic} in a mechanical manner for a couple of iterations and see what is the family of problems that emerges. Only after that we try to relate it to natural graph problems, such as $k$-matchings.
    \item We keep the size of the output alphabet manageable by collapsing multiple distinct labels to one label. However, we allow for some new ``unnatural'' labels (here: $\X$) that emerge in the process in addition to the ``natural'' labels that we had in the original problem (here: $\M,\P,\O$).
\end{itemize}
While many of the elements in our proof have some flexibility---we often somewhat liberally ``round down'' and simplify problems in order to keep the proofs and definitions as simple as possible---we are not aware of any way of generalizing \eqref{eq:mm} to something similar to \eqref{eq:pixyd} so that we could use a natural $3$-symbol alphabet and still prove a speedup result.

Finally, we point out that the new label $\X$ that emerged in the mechanical process was not exactly a wildcard symbol. It behaved a bit like a wildcard in certain contexts, and we simply forced this simple interpretation by relaxing the problem and allowing it to behave like a wildcard in all contexts.

\section{Lower bounds in the LOCAL model}\label{sec:randomized}

In \sectionref{sec:deterministic} we gave a linear-in-$\Delta$ lower bound for deterministic distributed algorithms in the port-numbering model. Now we would like to
\begin{enumerate}
    \item extend the result from the port-numbering model to the LOCAL model,
    \item extend the result from deterministic algorithms to randomized algorithms.
\end{enumerate}

\subsection{LOCAL model}

In the area of distributed graph algorithms, there are two widely used models of computing: \emph{LOCAL} and \emph{CONGEST}~\cite{Peleg2000}. The LOCAL model is strictly stronger than the CONGEST model; as we are interested in lower bounds, we will here use the LOCAL model to make our result as widely applicable as possible.

In the LOCAL model, each node is labeled with a \emph{unique identifier}. If there are $n$ nodes in the network, the unique identifiers are some subset of $\{1,2,\dotsc,\poly(n)\}$. Put otherwise, the unique identifiers are $O(\log n)$-bit natural numbers. We assume that the nodes know $n$, which can be done without loss of generality as we are proving impossibility results.

Other than the assumption of unique identifiers, the model of computing is the same as what we already used in \sectionref{sec:deterministic}: nodes are computational entities, edges are communication links, and computation proceeds in synchronous communication rounds. The computational power of each entity and the bandwidth is not limited, that is, each node can send messages of arbitrary size and can perform local computation of arbitrary complexity. Initially, each node knows its own unique identifier. In each round, each node can:
\begin{itemize}
    \item send an arbitrary message to each of its neighbors;
    \item receive a message from each neighbor;
    \item perform some local computation and update its own state.
\end{itemize}
A distributed algorithm runs in $T$ rounds in the LOCAL model if after $T$ rounds each node has terminated and outputted a local output, and these local outputs constitute a feasible global solution. Since the communication power is not limited, a distributed algorithm that runs in $T$ rounds in the LOCAL model can be seen as a procedure that does the following at each node $v$:
\begin{enumerate}
    \item gather the topology of the network up to distance $T$ form $v$; 
    \item perform unbounded local computation to produce its local output.
\end{enumerate}

In a \emph{randomized} algorithm the nodes are labeled also with a stream of random bits. We are primarily interested in Monte Carlo algorithms that find a maximal matching with high probability, which we here define so that the running time is fixed and the global success probability is at least $1-1/n$.

\subsection{High-level plan}

We will use a construction similar to \sectionref{sec:deterministic}. In particular, we will prove a lower bound for the problem of finding a maximal matching in a $2$-colored regular tree. Naturally, the same lower bound then holds also in the general case, even if we do not have a $2$-coloring.

Ideally, we would like to first extend the proof so that it holds also with unique identifiers, this way derive a lower bound for \emph{deterministic} algorithms in the LOCAL model, and then later add random input bits to derive a lower bound for \emph{randomized} algorithms in the LOCAL model.

Unfortunately, the speedup simulation technique does not work well with unique identifiers. If we look at e.g.\ regions $D_1$ and $D_2$ in \figureref{fig:speedup}, the inputs in these parts are no longer independent: for example, if $D_1$ contains a node with unique identifier $1$, we know that $D_2$ cannot contain such a node.

Hence, as usual, we take a different route~\cite{Brandt2016}: we first add randomness. We repeat the analysis of \sectionref{sec:deterministic} for randomized Monte Carlo algorithms in the port-numbering model. This way we can still use independence: random bits in $D_2$ are independent of random bits in $D_1$.

Once we have a lower bound for Monte Carlo randomized algorithms in the port-numbering model, it is straightforward to turn this into a lower bound for Monte Carlo randomized algorithms in the LOCAL model. To see this, we first observe that the lower bound holds also even if all nodes are labeled with the value of $n$ (such extra information does not change anything in the proof). But now if we know $n$, it is easy to use randomness to construct $O(\log n)$-bit labels that are unique with high probability. Hence a Monte Carlo randomized algorithm in the LOCAL model can be turned into a Monte Carlo randomized algorithm in the port-numbering model (with knowledge of $n$), and our lower bound applies. Finally, a lower bound for randomized algorithms in the LOCAL model trivially implies a lower bound also for deterministic algorithms in the LOCAL model.

\subsection{White and black randomized algorithms}

Let us now extend the concepts of white and black algorithms to white and black \emph{randomized} algorithms. As discussed above, we use the port-numbering model augmented with randomness; in brief, each node is labeled with a random bit string, and in a time-$T$ algorithm, the output of a node $u$ may depend on the random bit strings of all nodes $v$ that are within distance $T$ from $u$.

We say that $A$ is a \emph{white randomized algorithm} for $\Pi = (W,B)$ with a \emph{local error probability} $p$ if the following holds:
\begin{enumerate}
    \item White nodes produce labels for the incident edges, and black nodes produce an empty output.
    \item For each white node, the labels of the incident edges always form a word in $W$.
    \item For each black node, the labels of the incident edges form a word in $B$ with probability at least $1-p$.
\end{enumerate}
A \emph{black randomized algorithm} is analogous, with the roles of white and black nodes reversed.

Note that if we are given any Monte Carlo randomized algorithm $A$ that solves $\Pi$ with high probability in time $T$, we can always turn it into a white randomized algorithm (or black randomized algorithm) $A'$ that solves $\Pi$ with a local error probability $1/n$ and running time $T+O(1)$: the local failure probability cannot exceed the global failure probability, and if some white nodes are unhappy (violating constraints in $W$), we can locally fix it so that all white nodes are happy and only black nodes are unhappy (e.g.\ all unhappy white nodes simply pick the first word of $W$). Hence it is sufficient to prove a lower bound for white randomized algorithms.

\subsection{Speedup simulation}

We will first prove a probabilistic version of \lemmaref{lem:speedup}; the proof follows the same strategy as~\cite{Brandt2016}:

\begin{lemma}\label{lem:prob-speedup}
    Assume that there exists a white randomized algorithm that solves $\Pi_\Delta(x,y)$ in $T \ge 1$ rounds in trees, for $\Delta \ge 2x+y+1$, with local error probability $p \le 1/4^{\Delta+1}$. Then there exists a white randomized algorithm that solves $\Pi_\Delta(x+1,y+x)$ in $T-1$ rounds in trees, with local error probability at most $q = 5 \Delta p^{\frac{1}{\Delta+1}}$.
\end{lemma}

To prove the lemma, assume that $A$ is the white randomized algorithm with running time $T$ and local error probability~$p$. Let $\alpha \in[0,1]$ be a parameter that we will fix later.

\subsubsection{Preliminaries}

Let $N(x,r)$ denote the radius-$r$ neighborhood of node $x$. Throughout this section, we will use $u, u', u_i$ etc.\ to refer to black nodes and $v, v', v_i$ etc.\ to refer to white nodes. For a black node $u$ and a white node $v$ adjacent to~$u$, define
\begin{align*}
U(u) &= N(u,T-1), \\
V(v) &= N(v,T), \\
D(u,v) &= V(v) \setminus U(u).
\end{align*}
With this notation, the regions in \figureref{fig:speedup} are
\[
U = U(u), \quad
V_i = V(v_i), \quad
D_i = D(u,v_i).
\]

It will be convenient to assume that the \emph{port numbering is chosen independently and uniformly at random}. If the local error probability is at most $p$ for an adversarial choice of port numbering, it is at most $p$ also for a random port numbering. In what follows, e.g.\ $U(u)$ refers to all random choices within distance $T-1$ from $u$: both the random bit string given to the nodes and the randomly chosen port numbers.

\begin{definition}[typical labels]
    We say that $x$ is a \emph{typical} label for edge $\{u,v\}$ w.r.t.\ $U(u)$ if
    \[
    \Pr_{D(u,v)}\bigl[ \text{$A$ labels $\{u,v\}$ with $x$} \bigm| U(u) \bigr] \ge \alpha.
    \]
    Otherwise the label is \emph{atypical}.
\end{definition}

Note that as long as $\alpha \le 1/4$, there is always at least one typical label for each edge (recall that the alphabet size is~$4$).

\begin{definition}[lucky neighborhoods]
    We say that $U(u)$ is \emph{lucky} if the following holds: for each edge $\{u,v\}$ we can pick any typical label and this makes node $u$ happy, i.e., any combination of typical labels forms a word in $B_\Delta(x,y)$. Otherwise, we call $U(u)$ \emph{unlucky}.
\end{definition}

\begin{lemma}
    The probability that $U(u)$ is unlucky is at most $p/\alpha^\Delta$.
\end{lemma}
\begin{proof}
    If $U(u)$ is unlucky, then for each edge $e_i = \{u,v_i\}$ we can choose a typical edge label $x_i$ so that $u$ is unhappy. By definition, given $U(u)$, algorithm $A$ will output $x_i$ on edge $e_i$ with probability at least $\alpha$, and this only depends on $D(u,v_i)$. As the regions $D(u,v_i)$ are independent, we see that $u$ is unhappy with probability at least $\alpha^\Delta$.
    
    Let $\beta$ be the probability that $U(u)$ is unlucky. Node $u$ is then unhappy with probability at least $\beta \alpha^\Delta$, which is at most $p$ by assumption.
\end{proof}

\begin{definition}[nice neighborhoods]
    We say that $V(v)$ is \emph{nice} if the following holds for each edge $e = \{v, u\}$ incident to $v$: the label assigned by algorithm $A$ to $e$ is a typical label for $\{u,v\}$ w.r.t.\ $U(u)$. Otherwise $V(v)$ is \emph{bad}.
\end{definition}

\begin{lemma}
    The probability that $V(v)$ is bad is at most $4\Delta\alpha$.
\end{lemma}
\begin{proof}
    There are $\Delta$ edges incident to $v$; let us focus on one of them, say $e = \{v, u\}$. We next analyze the probability that $e$ has an atypical label.
    
    We first fix $U(u)$. Then we see which labels are atypical for $e$ w.r.t.\ $U(u)$. Now consider each possible label $x$; if $x$ is atypical, then the probability that $A$ outputs $x$ for $e$ given $U(u)$ is less than $\alpha$. Summing over all possible atypical labels (at most $4$) and all possible choices of $U(u)$, the probability that the output of $A$ for $e$ is atypical is at most $4\alpha$.
    
    The claim follows by union bound.
\end{proof}

\begin{definition}[friendly neighborhoods]
    We say that $V(v)$ is \emph{friendly} if:
    \begin{enumerate}
        \item $V(v)$ is nice,
        \item $U(u)$ is lucky for each black neighbor $u$ of $v$.
    \end{enumerate}
    Otherwise $V(v)$ is \emph{unfriendly}.
\end{definition}

\begin{lemma}
    The probability that $V(v)$ is unfriendly is at most $\Delta(4\alpha + p/\alpha^\Delta)$.
\end{lemma}
\begin{proof}
    There are $\Delta+1$ bad events: one possible bad neighborhood, with probability at most $4\Delta\alpha$, and $\Delta$ possible unlucky neighborhoods, each with probability at most $p/\alpha^\Delta$. By union bound, we have an unfriendly neighborhood with probability at most $\Delta(4\alpha + p/\alpha^\Delta)$.
\end{proof}

Now if we choose $\alpha = p^{\frac{1}{\Delta+1}} \le 1/4$, we obtain:
\begin{corollary}
    The probability that $V(v)$ is unfriendly is at most $q = 5 \Delta p^{\frac{1}{\Delta+1}}$.
\end{corollary}

\subsubsection{Speedup simulation for friendly neighborhoods}

We will first look at speedup simulation in friendly neighborhoods. The high-level plan is this:
\begin{itemize}
    \item Speedup simulation is well-defined for each black node $u$ such that $U(u)$ is lucky.
    \item Speedup simulation results in a good output for each white node $v$ such that $V(v)$ is friendly.
\end{itemize}

Using the notation of \sectionref{ssec:speedup}, we redefine $A_1$ as follows:
\begin{framed}
    \noindent Set $S_i$ consists of all typical labels for $e_i$ w.r.t.\ $U(u)$.
\end{framed}
Let us now look at how to extend algorithms $A_2$ to $A_4$ and their analysis to the probabilistic setting, assuming a friendly neighborhood:
\begin{itemize}
    \item Algorithm $A_2$: The mapping from $15$ sets to $6$ labels is identical to the deterministic version.
    \item Output of $A_2$: The analysis holds verbatim, as we assumed that $u$ is lucky and hence any combination of typical labels has to make $u$ happy.
    \item Algorithm $A_3$: The mapping from $6$ sets to $4$ labels is identical to the deterministic version.
    \item Algorithm $A_4$: The mapping from $4$ labels to $4$ labels is identical to the deterministic version.
    \item Output of $A_4$: As we look at a white node in a friendly neighborhood, we know that the output of $A$ is typical, and hence the output of $A$ is contained in the sets of typical labels that $A_1$ outputs. The rest of the analysis holds verbatim.
\end{itemize}

\subsubsection{Speedup simulation for all neighborhoods}

Let us now address the case of unfriendly neighborhoods. We modify $A_4$ so that if node $u$ is unlucky, we produce some arbitrary fixed output from $W_\Delta(x+1,y+x)$. This way $A_4$ is always well-defined and the output of a black node is always in $W_\Delta(x+1,y+x)$.

Furthermore, we know that the labels incident to a white node $v$ are in $B_\Delta(x+1,y+x)$ whenever $V(v)$ is friendly, and this happens with probability at least $1-q$. We conclude that $A_4$ is a black randomized algorithm that solves $\Pi'_\Delta(x+1,y+x)$ with local error probability at most $q$.

Finally, we obtain a white randomized algorithm for $\Pi_\Delta(x+1,y+x)$ by reversing the roles of white and black nodes. This concludes the proof of \lemmaref{lem:prob-speedup}.

\subsection{Multiple speedup steps}

By a repeated application of \lemmaref{lem:prob-speedup}, we obtain the following corollary that is analogous to \corollaryref{cor:speedup}:

\begin{corollary}\label{cor:prob-speedup}
    Assume that there exists a white randomized algorithm that solves $\Pi_\Delta(x,y)$ in $T$ rounds in trees with local error probability $p$, for $\Delta \ge x+y+T(x+1+(T-1)/2)$. Then there is a white randomized algorithm that solves $\Pi_\Delta(x',y')$ in $0$ rounds in trees with local error probability $p'$ for
    \begin{align*}
    x' &= x + T, \\
    y' &= y + T\bigl(x + (T-1)/2\bigr), \\
    p' &\le (5\Delta)^2 p^{1/(\Delta+1)^T}.
    \end{align*}
\end{corollary}

\begin{proof}
    Let $p_i$ be the local error probability after applying \lemmaref{lem:prob-speedup} for $i$ times, with $p_0 = p$ and $p_T = p'$. We prove by induction that
    \[
    p_i \le (5\Delta)^2 p^{1/(\Delta+1)^i}.
    \]
    The base case is $p_0 \le (5\Delta)^2 p$, which trivially holds. For the inductive step, note that
    \[
    \begin{split}
    p_{i+1}
    &\le 5 \Delta \cdot p_i^{1/(\Delta+1)} \\
    &\le 5 \Delta \cdot \bigl((5\Delta)^2 p^{1/(\Delta+1)^i}\bigr)^{1/(\Delta+1)} \\
    &\le (5\Delta)^2 \cdot p^{1/(\Delta+1)^{i+1}}. \qedhere
    \end{split}
    \]
\end{proof}

\subsection{Base case}

\begin{lemma}\label{lem:prob-base}
    There is no white randomized algorithm that solves $\Pi_\Delta(x,y)$, for $x,y \le \Delta/8$ and $\Delta\ge 8$, in $0$ rounds in trees with local error probability $p \le 1/{\Delta^\Delta}$.
\end{lemma}

\begin{proof}
    Fix a white randomized algorithm $A$ that runs in $0$ rounds. As each white node has no knowledge of their neighborhood, they all have the same probability $q$ to choose the output $\M\s\O^{\Delta-x-1}X^x$, and probability $1-q$ to output $\P^{\Delta-x-y}\s \O^y\s\X^x$. There are two cases:
    \begin{itemize}
        \item $q\ge\frac12$: Consider a black node $u$ with $\Delta$ white neighbors $v_1, v_2, \dotsc, v_\Delta$. As the port numbering of $v_i$ is chosen uniformly at random, the probability that $v_i$ labels the edge $\{v_i, u\}$ with $\M$ is $q/\Delta \ge 1/(2\Delta)$. Now if the first $x+2$ white neighbors output $\M$ on the connecting edge, node $u$ will be unhappy, and this happens with probability at least
        \[
        \frac{1}{(2\Delta)^{x+2}}
        \ge \frac{1}{(2\Delta)^{\Delta/8+2}}
        \ge \frac{1}{\Delta^\Delta}.
        \]
        \item $q\le\frac12$: Now a black node $u$ is unhappy if all white neighbors produce an output of type $\P^{\Delta-x-y}\s \O^y\s\X^x$, and furthermore the first $x+y+1$ neighbors output $\P$ on the connecting edge; this happens with probability at least
        \[
        \frac{1}{2^\Delta} \cdot \frac{1}{\Delta^{x+y+1}}
        \ge \frac{1}{2^\Delta} \cdot \frac{1}{\Delta^{\Delta/4+1}}
        \ge \frac{1}{\Delta^\Delta}.
        \qedhere
        \]
    \end{itemize}
\end{proof}

\subsection{Putting things together}

\begin{lemma}\label{lem:prob-multi}
    For a sufficiently large $\Delta$, there is no white algorithm that solves $\Pi_\Delta(x,0)$ for $x \le \sqrt{\Delta}$ in $T \le \sqrt{\Delta}/16$ rounds in trees with local error probability $p < 2^{-\Delta^{2T+2}}$.
\end{lemma}
\begin{proof}
    Assume that $A$ solves $\Pi_\Delta(x,0)$ in $T \le \sqrt{\Delta}/16$ rounds with local error probability $p$. By \corollaryref{cor:prob-speedup}, we obtain a white randomized algorithm $A'$ that solves $\Pi_\Delta(x',y')$ in $0$ rounds with local error probability $p'$, where
    \begin{align*}
    x' &= x + T \le \frac{17}{16} \sqrt{\Delta} \le \frac{1}{8} \Delta, \\
    y' &= y + T\bigl(x + (T-1)/2\bigr) \le \frac{33}{512}\Delta \le \frac{1}{8} \Delta, \\
    p' &\le (5\Delta)^2 p^{1/(\Delta+1)^T}.
    \end{align*}
    By \lemmaref{lem:prob-base} we have $p' > 1/\Delta^\Delta$ and therefore
    \[
    p
    > \frac{1}{\bigl((5\Delta)^2 \Delta^\Delta\bigr)^{(\Delta+1)^T}}
    > \frac{1}{2^{\Delta^{2T+2}}}.
    \qedhere
    \]
\end{proof}

\begin{theorem}\label{thm:prob-k-matching}
    There exists a constant $c>0$ such that for any $n$ and $\Delta$ large enough satisfying that at least one of $n$ and $\Delta$ is even and $\Delta < n/10$,
    any algorithm that finds a $\sqrt{\Delta}$-matching in $\Delta$-regular graphs of $n$ nodes with probability at least $1-1/n$ requires at least $c \min\{\sqrt{\Delta},\log_\Delta \log n\}$ rounds.
\end{theorem}
\begin{proof}
    Let us assume for a contradiction that for any constant $c > 0$ and any positive integers $n_0$ and $\Delta_0$ there exist $n > n_0$ and $\Delta > \Delta_0$ satisfying $\Delta < n/10$ such that there is an algorithm that finds a $\sqrt{\Delta}$-matching in $\Delta$-regular graphs of $n$ nodes in less than $c \min\{\sqrt{\Delta},\log_\Delta \log n\}$ rounds with probability at least $1-1/n$.
    
    By \lemmaref{lem:k-matching-pi} there is then a white randomized algorithm $A$ that solves $\Pi_\Delta(\sqrt{\Delta},0)$ with only an additional constant overhead. Let $\bar{c}$ be the largest $c \le 0.1$ such that the obtained algorithm $A$ runs in $T < \sqrt{\Delta}/16$ rounds.

    We fix $c$ to be the minimum of $\bar{c}$ and the constant $c$ of Lemma~\ref{lem:regulargraphs}. Then, we fix $\Delta_0$ as the minimum value required to apply \lemmaref{lem:prob-multi}, and $n_0$ as the minimum value required to apply Lemma~\ref{lem:regulargraphs}. We then apply Lemma~\ref{lem:regulargraphs} to construct a $\Delta$-regular graph of $n$ nodes that contains a node $v$ such that the radius-$t$ neighborhood of $v$ is isomorphic to a $\Delta$-regular tree, for some $t \ge c \log_\Delta n$.
    Since $T < c \log_\Delta \log n \le t$ and $T < \sqrt{\Delta}/16$, by \lemmaref{lem:prob-multi}, algorithm $A$ fails in such a neighborhood with probability at least
    \[
    \frac{1}{2^{\Delta^{2T+2}}}
    > \frac{1}{2^{\Delta^{\log_\Delta \log n}}}
    = \frac{1}{n}.
    \]
    Hence the global success probability cannot be $1-1/n$.
\end{proof}

Now the same idea as in \theoremref{thm:full} gives our main result:

\begin{theorem}\label{thm:prob-full}
    There exists a constant $c>0$ such that for any $n$ and $\Delta$ large enough satisfying that at least one of $n$ and $\Delta$ is even and $\Delta < n/10$,
    any LOCAL-model algorithm that finds a maximal matching in $\Delta$-regular graphs of $n$ nodes with probability at least $1-1/n$ requires at least $c \min\{\Delta,\log_\Delta \log n\}$ rounds.
\end{theorem}

Note that $\Delta$ and $\log_\Delta \log n$ become roughly the same by setting $\Delta \approx \log \log n / \log \log \log n$. Hence, if we consider $\Delta$ to be an upper bound of the maximum degree, we directly get the following.

\begin{corollary}\label{cor:local-maxdeg}
    Any LOCAL-model algorithm that finds a maximal matching in graphs of maximum degree at most $\Delta$ with probability at least $1-1/n$ requires $\Omega(\min\{\Delta,\log \log n / \log \log \log n\})$ rounds.
\end{corollary}

The following corollary gives a lower bound for the running time of any algorithm for maximal matching, if we express it solely as a function of $n$.

\begin{corollary}\label{cor:local-maxn}
    Any LOCAL-model algorithm that finds a maximal matching in $\Delta$-regular graphs with probability at least $1-1/n$ requires $\Omega(\log \log n / \log \log \log n)$ rounds.
\end{corollary}

\subsection{Deterministic lower bound via speedup simulation}

Our randomized lower bound implies also a stronger deterministic lower bound. The proof is again based on a speedup argument: we show that the existence of a (too) fast deterministic algorithm implies the existence of an even faster deterministic algorithm. Using the simple observation that the randomized complexity of a problem is at most its deterministic complexity, we will then obtain a contradiction with \theoremref{thm:prob-full}. Our proof is similar to \citet[Theorem 6]{chang16exponential}, and the coloring algorithm follows the idea of \citet[Remark 3.6]{Barenboim2016}.

\begin{theorem}\label{thm:det-full}
    There exists a constant $c>0$ such that for any $n$ and $\Delta$ large enough satisfying that at least one of $n$ and $\Delta$ is even and $\Delta < n/10$,
    any deterministic LOCAL-model algorithm that finds a maximal matching in $\Delta$-regular graphs of $n$ nodes requires at least $c \min\{\Delta,\log_\Delta n\}$ rounds.
\end{theorem}
\begin{proof}
    Let us assume for a contradiction that for any constant $c > 0$ and any positive integers $N_0$ and $\Delta_0$ there exist $N > N_0$ and $\Delta > \Delta_0$ satisfying $\Delta < N/10$ such that there is a deterministic algorithm $A$ that finds a maximal matching in $\Delta$-regular graphs of $N$ nodes in less than $T = c \min\{\Delta,\log_\Delta N\}$ rounds. We use this assumption to create a randomized algorithm that violates Theorem~\ref{thm:prob-full}.
    
    Let $\hat{c}$, $N_0$, and $\Delta_0$ be, respectively, the constant $c$, the minimum size $n$, and the minimum degree $\Delta$, required to apply Theorem~\ref{thm:prob-full}. Let $\bar{c}$ be the minimum of $0.1$ and $\hat{c}$. Let $c = \epsilon \bar{c}$ for some positive constant $\epsilon$ to be fixed later.
    We get that there exists an algorithm $A$ that finds a maximal matching in $\Delta$-regular graphs of $N$ nodes in less than $T = c \min\{\Delta,\log_\Delta N\}$ rounds, for some specific $N > N_0$ and $\Delta > \Delta_0$. We show how to construct an algorithm $A'$ that is able to find a maximal matching in any $\Delta$-regular graphs of $n = 2^N$ nodes in less than $\epsilon^{-1} T$ rounds. Note that $T = c \min\{\Delta,\log_\Delta N\} = c \min\{\Delta,\log_\Delta \log n\}$. This implies that we can solve maximal matching on $\Delta$-regular graphs ($\Delta > \Delta_0$) of size $n = 2^N > N_0$ in $\epsilon^{-1}T =  \bar{c} \min\{\Delta,\log_\Delta \log n\}$, a contradiction with Theorem~\ref{thm:prob-full}.
    
    Let $G$ be any $\Delta$-regular graph of $n$ nodes. We show how to simulate $A$ on $G$, for any possible given assignment of unique identifiers in $\{1,\ldots,\poly(n)\}$ for the nodes of $G$. In order to simulate $A$ we need to compute a labeling that \emph{locally} looks like an assignment of unique identifiers of size $N$, and we need to make sure that each neighborhood contains less than $N$ nodes.   
    We compute an $N$-coloring of $G^{2T+2}$, the $(2T+2)$th power of $G$, using three iterations of Linial's coloring algorithm~\cite[Corollary 4.1]{Linial1992}. This algorithm can be used to color a $k$-colored graph of maximum degree $\bar \Delta$ with e.g.
    \[
    O\bigl(\bar \Delta^2 \bigl(\log \log \log k + \log \bar \Delta\bigr)\bigr)
    \]
    colors in three rounds. The power graph $G^{2T+2}$ has maximum degree $\bar \Delta \leq \Delta^{2T+2}$ and thus we get a coloring of $G^{2T+2}$ with
    \[
    O\bigl(\Delta^{2(2T+2)} \bigl(\log \log N + \log \Delta^{2T+2}\bigr)\bigr)
    \]
    colors. Since $T \le c \log_\Delta N$, we see that each $T$-radius neighborhood contains $O(\Delta^T) < N$ nodes and that the number of colors is less than $N$.  The coloring can be computed in time $O(T)$ in the original graph $G$.

    Finally we simulate $A$ on the computed coloring. If we look at the radius-$(T+1)$ neighborhood of any node in $G$, the number of nodes is less than $N$, and the coloring looks like an assignment of unique identifiers from $\{1,2,\dotsc,N\}$, and thus the output of $A$ is well defined and correct by a standard argument---see e.g.\ \cite[Theorem 6]{chang16exponential}. Furthermore, the simulation of $A$ on $G$ can be done in $T$ rounds. Our simulation runs in $O(T)$ rounds on total.
    
    We can easily transform this deterministic algorithm into a randomized algorithm with the same runtime on the same graph class: in the beginning, each node simply picks a random value from $\{1,2,\dotsc,n^3\}$, interprets its value as an identifier, and then runs the deterministic simulation with these identifiers. The picked values are globally unique with probability at least $1-1/n$, guaranteeing that the randomized algorithm is correct w.h.p.
    As previously stated, the obtained running time gives a contradiction with \theoremref{thm:prob-full}.
\end{proof}

By setting $\Delta \approx \log n / \log \log n$ in Theorem~\ref{thm:det-full} we directly get the following.
\begin{corollary}\label{cor:local-det-maxn}
    Any deterministic LOCAL-model algorithm that finds a maximal matching in $\Delta$-regular graphs requires $\Omega(\log n / \log \log n)$ rounds.
\end{corollary}

If we consider $\Delta$ to be an upper bound of the maximum degree, we directly get the following.
\begin{corollary}\label{cor:local-det-maxdeg}
    Any deterministic LOCAL-model algorithm that finds a maximal matching in graphs of maximum degree $\Delta$ requires $\Omega(\min\{\Delta, \log n / \log \log n\})$ rounds.
\end{corollary}

\subsection{Lower bounds for MIS}

If we have any algorithm that finds an MIS, we can simulate it in the line graph and obtain an algorithm for MM; the simulation overhead is constant and the increase in the maximum degree is also bounded by a constant factor. We obtain:

\begin{corollary}\label{cor:prob-mis}
    Any LOCAL-model algorithm that finds a maximal independent set in graphs of maximum degree $\Delta$ with probability at least $1-1/n$ requires $\Omega(\min\{\Delta, \log \log n / \log \log \log n\})$ rounds.
\end{corollary}

\begin{corollary}\label{cor:det-mis}
    Any deterministic LOCAL-model algorithm that finds a maximal independent set in graphs of maximum degree $\Delta$ requires $\Omega(\min\{\Delta, \log n / \log \log n\})$ rounds.
\end{corollary}

\subsection{Matching upper bounds}

Recall the bound from \theoremref{thm:prob-full}: it excludes the possibility of randomized distributed algorithms that find a maximal matching in $o(\Delta) + o(\log \log n / \log \log \log n)$ rounds in general. The first term is optimal in general: there are deterministic and randomized algorithms that find a maximal matching in $O(\Delta) + o(\log \log n / \log \log \log n)$ rounds~\cite{panconesi01simple}. The second term is optimal for randomized MM algorithms in trees: there is a randomized algorithm that finds a maximal matching in trees in $o(\Delta) + O(\log \log n / \log \log \log n)$ rounds \cite[Theorem 7.3]{Barenboim2016}. Note that the construction of our lower bounds ensures that they already hold in trees.

\section{Discussion}

We have learned that the prior algorithms for MM and MIS with a running time of $O(\Delta + \log^* n)$ rounds~\cite{panconesi01simple,barenboim14distributed} are optimal for a wide range of parameters: in order to solve MM or MIS in time $o(\Delta) + g(n)$, we must increase $g(n)$ from $O(\log^* n)$ all the way close to $O(\log \log n)$.

A priori, one might have expected that we should be able to achieve a smooth transition for MM algorithms between $O(\Delta + \log^* n)$ and $O(\log \Delta + \log^c \log n)$, which are the currently best known complexities as a function of $n$ and $\Delta$, respectively. However, this turned out to be not the case. We conjecture that the qualitative gap between $\log^* n$ and $\log^c \log n$ regions is closely related to similar gaps in the landscape of locally checkable labeling problems~\cite{chang16exponential}.

\subsection{Open questions}

After the preliminary conference version of this work, one of the key open questions was related to the complexity of maximal independent set in the region $\Delta \gg \log \log n$, but thanks to the recent work by \citet{Rozhon2020}, this question is now largely settled. We now know that the maximal independent set problem cannot be much harder to solve than the maximal matching problem.

However, the complexity of vertex coloring and edge coloring is still wide open. Here is one concrete example of a big open question, closely related to the theme of the current work: Is it possible to find a $(\Delta+1)$-vertex coloring in $O(\log \Delta + \log^* n)$ rounds, for all $\Delta$?

Several more fine-grained questions related to the complexity of MM still remain: For example, we proved that MM requires $\Omega(\min\{\Delta,\log \log n / \log \log \log n\})$ rounds in trees, even for randomized algorithms, and the dependency on $n$ is also tight in trees~\cite{Barenboim2016}. In general graphs the current randomized algorithms take $O(\log \Delta + \log^c \log n)$ rounds for $c > 1$. A natural open question is whether or not MM in general graphs is strictly harder than MM in trees, if we allow the dependency on $\Delta$ to be $O(\log \Delta)$. Proving such a separation through the speedup simulation technique would require one to be able to handle cycles, which is an interesting open question by itself. Also, when we consider algorithms that take $\poly \log \log n$ rounds as a function of $n$, there is a gap in the complexity as a function of $\Delta$: the current lower bound is $\Omega(\log \Delta / \log \log \Delta)$, while the upper bound is $O(\log \Delta)$; it is wide open whether the speedup simulation technique is applicable in the study of such questions.

\subsection{Recent follow-up work}

Computers played a key role in the present work especially in the discovery of an appropriate problem family \eqref{eq:pixyd}: speedup simulation is a fairly laborious but mechanical process and hence well suited to be automated. The computer program that we used to help us with this project has been now further developed into a freely available open-source web application known as Round Eliminator~\cite{Olivetti2019}, and it has already been used to assist in the discovery of several other lower and upper bounds, see e.g.~\cite{Balliu2019a,Brandt2020trulytight,Balliu2020rulingset}.

Recent work has also led to an even more fine-grained characterization of the distributed computational complexity of maximal matchings. Recall the trivial algorithm for bipartite maximal matching from Section~\ref{sec:deterministic}; one can verify that the round complexity of the algorithm is exactly $2\Delta-1$. While the present work showed that this is asymptotically optimal, the latest follow-up work~\cite{Brandt2020trulytight} showed that this is exactly optimal: the problem cannot be solved in $2\Delta-2$ rounds.

\begin{acks}
This work is an extended and revised version of a preliminary conference report that appeared in the 60th Annual IEEE Symposium on Foundations of Computer Science (FOCS 2019).

We would like to thank Mohsen Ghaffari for pointing out that the randomized lower bound implies a better deterministic lower bound, and Seth Pettie for pointing out that the bound is tight for trees, as well as for many other helpful comments. Many thanks also to Laurent Feuilloley and Tuomo Lempi\"ainen for comments and discussions, and to the reviewers of the previous versions of this work for their numerous helpful suggestions.

This work was supported in part by the Academy of Finland, Grants 285721 and 314888 and the Project ESTATE (ANR-16-CE25-0009-03).
\end{acks}

\bibliographystyle{ACM-Reference-Format}
\bibliography{maximal-matching}

\end{document}